\documentclass[a4paper,USenglish,cleveref, autoref, thm-restate,authorcolumns]{lipics-v2019}

\nolinenumbers

\usepackage{adjustbox}
\usepackage{amsmath,amsfonts,amsthm}
\usepackage{algorithm}
\usepackage[noend]{algpseudocode}
\usepackage{booktabs}
\usepackage{csquotes}
\usepackage{cite}
\usepackage{enumerate}
\usepackage{multirow}
\usepackage[autolanguage]{numprint}
\usepackage{graphicx}
\usepackage{hyperref}
\usepackage[utf8]{inputenc}
\usepackage[locale=US]{siunitx}
\usepackage{subcaption}
\usepackage{textcomp}
\usepackage{url}
\usepackage{tikz}
\usepackage{xspace}

\definecolor{light-gray}{gray}{0.93}

\newcommand{\ie}{i.\,e.,\xspace}
\newcommand{\eg}{e.\,g.,\xspace}
\newcommand{\wrt}{w.\,r.\,t.\xspace}

\newcommand{\etal}{et al.\xspace}
\newcommand{\quot}[1]{``#1''}

\newcommand{\Oh}{\ensuremath{\mathcal{O}}}

\newcommand{\dist}[1]{\mathrm{dist}(#1)}

\newcommand{\ceil}[1]{\left\lceil #1 \right\rceil}

\newcommand{\lonerel}{\ensuremath{\mathrm{L1}_\mathrm{rel}}}
\newcommand{\ltworel}{\ensuremath{\mathrm{L2}_\mathrm{rel}}}
\newcommand{\erel}{\ensuremath{\mathrm{E}_\mathrm{rel}}}
\newcommand{\onen}[1]{\lVert #1 \rVert_1}
\newcommand{\twon}[1]{\lVert #1 \rVert_2}

\newcommand{\diag}[1]{\operatorname{diag}(#1)}
\newcommand{\diam}[1]{\operatorname{diam}(#1)}

\newcommand{\mat}[1]{\mathbf{#1}}
\newcommand{\myvec}[1]{\mathbf{#1}}
\newcommand{\ment}[3]{\mathbf{#1}[#2,#3]}
\newcommand{\vent}[2]{\myvec{#1}[#2]}
\newcommand{\Lpinv}{\mat{L}^\dagger}
\newcommand{\effres}[2]{\myvec{r}(#1,#2)}
\newcommand{\trace}[1]{\operatorname{tr}(#1)}
\newcommand{\uvec}[1]{\mathbf{e}_{#1}}
\newcommand{\onesvec}{\mathbf{1}}
\newcommand{\ecc}{\operatorname{ecc}}
\newcommand{\farnel}[2]{f^{el}(#2)}
\newcommand{\farnc}[2]{f^{c}(#2)}

\newcommand{\poly}[1]{\operatorname{poly}(#1)}
\newcommand{\numtrees}[4]{N_{#1,#4}(#2,#3)}

\newcommand{\bekasH}{\tool{Bekas-h}\xspace}
\newcommand{\bekas}{\tool{Bekas}\xspace}
\newcommand{\jltLamg}{\tool{Lamg-jlt}\xspace}
\newcommand{\jltKyng}{\tool{Julia-jlt}\xspace}
\newcommand{\ust}{\tool{UST}\xspace}

\newcommand{\tool}[1]{\textsf{#1}}
\newcommand{\nwk}{\tool{NetworKit}\xspace}

\newcommand{\instTabColSep}{3pt}

\newcommand{\parallelScalability}{
\begin{figure}
\centering
\begin{subfigure}[t]{.5\columnwidth}
\centering
\includegraphics{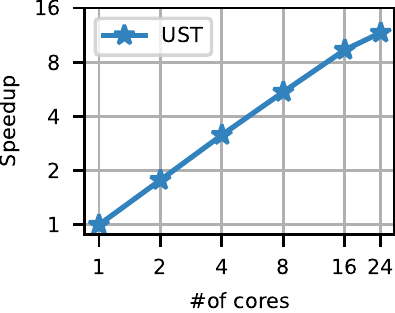}
\caption{Geometric mean of the speedup of \ust on multiple cores (shared memory) \wrt a sequential run.
Data points are aggregated over the instances of
Tables~\ref{tab:networks-medium-gt} and~\ref{tab:networks-medium}
\ifthenelse{\boolean{confversion}}{(see full version~\cite{angriman2020approximation})}{}.}
\label{fig:shared-mem-scalability}
\end{subfigure}\hfill
\begin{subfigure}[t]{.5\columnwidth}
\centering
\includegraphics{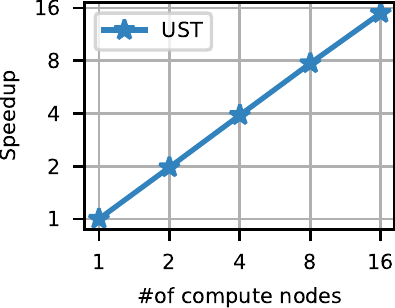}
\caption{Geometric mean of the speedup of \ust on multiple compute nodes \wrt \ust on a single compute node ($1\times 24$ cores).
Data points are aggregated over the instances of
Tables~\ref{tab:networks-medium-gt},~\ref{tab:networks-medium}, and~\ref{tab:networks-large}
\ifthenelse{\boolean{confversion}}{(see full version~\cite{angriman2020approximation})}{}.}
\label{fig:distr-mem-scalability}
\end{subfigure}\hfill
\caption{Parallel scalability of \ust ($\epsilon = 0.3$) with shared and with distributed memory.}
\label{fig:par-scal}
\end{figure}
}

\newcommand{\timeBreakdown}{
\begin{figure}
\centering
\begin{subfigure}[c]{.25\columnwidth}
\centering
\includegraphics{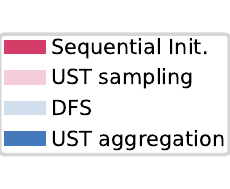}
\end{subfigure}\hfill
\begin{subfigure}[c]{.75\columnwidth}
\centering
\includegraphics{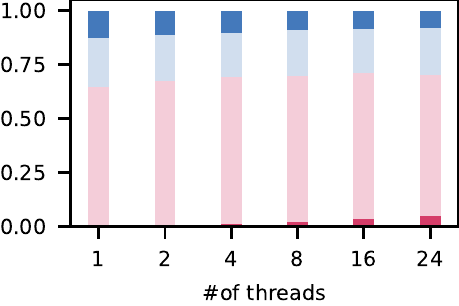}
\end{subfigure}
\caption{Breakdown of the running times of \ust with $\epsilon = 0.3$
\wrt \#of cores on $1\times 24$ cores. Data is aggregated with the geometric
mean over the instances of Tables~\ref{tab:networks-medium-gt}
and~\ref{tab:networks-medium} (see full version~\cite{angriman2020approximation}).}
\label{fig:breakdown}
\end{figure}
}

\newcommand{\syntheticInstances}{
\begin{figure}
\centering
\begin{subfigure}[t]{.5\textwidth}
\centering
\includegraphics{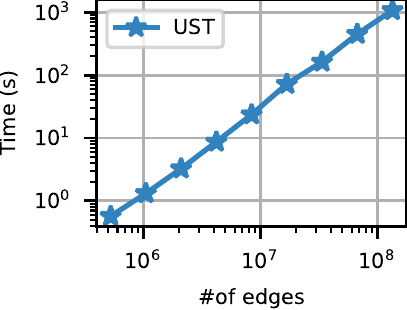}
\caption{Running time of \ust \wrt \#of edges.}
\label{fig:rmat-time}
\end{subfigure}\hfill
\begin{subfigure}[t]{.5\textwidth}
\centering
\includegraphics{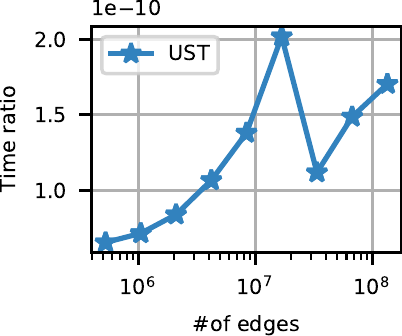}
\caption{Ratio of running time of \ust \wrt its theoretical running time
(see Theorem~\ref{thm:time-complexity}).}
\label{fig:rmat-ratio}
\end{subfigure}
\caption{Scalability of \ust on R-MAT graphs ($\epsilon = 0.3$, $1\times24$ cores).}
\label{fig:rmat}
\end{figure}
}

\newcommand{\lamgTol}{10^{-9}\xspace}
\newcommand{\shmemspeedup}{$\numprint{11.9}\times$\xspace}
\newcommand{\distrspeedup}{$\numprint{15.1}\times$\xspace}
\newcommand{\ustMaxAbsErr}{$\numprint{0.09}$\xspace}
\newcommand{\ustAvgAbsErr}{$\numprint{0.07}$\xspace}
\newcommand{\bekasMaxAbsErr}{$\numprint{2.43}$\xspace}
\newcommand{\bekasAvgAbsErr}{$\numprint{0.62}$\xspace}
\newcommand{\ustBekasSpeedup}{$\numprint{8.3}\times$\xspace}
\newcommand{\ustBekasHSpeedup}{$\numprint{25.4}\times$\xspace}
\newcommand{\bekasRanking}{$\numprint{4.3}\%$\xspace}
\newcommand{\ustRanking}{$\numprint{2.1}\%$\xspace}

\newcommand{\maxTimeHyp}{$\numprint{184}$\xspace}
\newcommand{\maxTimeRmat}{$\numprint{18}$\xspace}
\newcommand{\maxEdgesHyp}{$\numprint{83.9}$\xspace}
\newcommand{\maxEdgesRmat}{$\numprint{134.2}$\xspace}
\newcommand{\ustFracAllCores}{$\numprint{95.3}\%$\xspace}
\newcommand{\ustFracOneCore}{$\numprint{99.4}\%$\xspace}
\newcommand{\aggOverSampl}{$\numprint{2.2}\times$\xspace}
\newcommand{\avgAggregate}{$\numprint{31.2}\%$\xspace}
\newcommand{\avgSampling}{$\numprint{66.8}\%$\xspace}

\newcommand{\ustMemPR}{$\numprint{487.0}$\xspace}
\newcommand{\lamgMemPR}{$\numprint{1.6}$\xspace}

\usepackage{ifthen}
\newboolean{confversion}
\setboolean{confversion}{false}

\newcommand{\changed}[1]{\textcolor{blue}{#1}}

\title{Approximation of the Diagonal of a Laplacian's Pseudoinverse
for Complex Network Analysis}

\titlerunning{Approximation of $\diag{\Lpinv}$ for Complex Network Analysis}

\ArticleNo{17}

\author{Eugenio Angriman}{Department of Computer Science, Humboldt-Universit\"at zu Berlin, Germany}{angrimae@hu-berlin.de}{}{}

\author{Maria Predari}{Department of Computer Science, Humboldt-Universit\"at zu Berlin, Germany}{predarim@hu-berlin.de}{}{}

\author{Alexander van der Grinten}{Department of Computer Science, Humboldt-Universit\"at zu Berlin, Germany}{avdgrinten@hu-berlin.de}{}{}

\author{Henning Meyerhenke}{Department of Computer Science, Humboldt-Universit\"at zu Berlin, Germany}{meyerhenke@hu-berlin.de}{}{}

\authorrunning{E.\ Angriman, M.\ Predari, A.\ van der Grinten, H.\ Meyerhenke}

\Copyright{Eugenio Angriman, Maria Predari, Alexander van der Grinten,
Henning Meyerhenke}



\begin{document}

\sloppy

\maketitle

\begin{abstract}
The ubiquity of massive graph data sets in numerous applications requires fast algorithms for extracting
knowledge from these data.
We are motivated here by three electrical measures for the analysis of large small-world
graphs $G = (V, E)$ -- \ie graphs with diameter in $\mathcal{O}(\log |V|)$,
which are abundant in complex network analysis.
From a computational point of view, the three measures have in common that their crucial component is the
diagonal of the graph Laplacian's pseudoinverse, $\Lpinv$.
Computing $\diag{\Lpinv}$ exactly by pseudoinversion, however, is as expensive as dense matrix multiplication -- and the standard tools in practice even require
cubic time. Moreover, the pseudoinverse requires quadratic space -- hardly feasible for large graphs. Resorting to approximation by, \eg
using the Johnson-Lindenstrauss transform, requires the solution of $\Oh(\log |V| / \epsilon^2)$
Laplacian linear systems to guarantee a relative error, which is still very expensive for large inputs.

In this paper, we present a novel approximation algorithm
that requires the solution of only one Laplacian linear system.
The remaining parts are purely combinatorial --
mainly sampling uniform spanning trees, which we relate to $\diag{\Lpinv}$
via effective resistances. For small-world networks, our algorithm
obtains a $\pm \epsilon$-approximation with high probability, in a time that is
nearly-linear in $|E|$ and quadratic in $1 / \epsilon$.
Another positive aspect of our algorithm is its parallel nature due to independent sampling.
We thus provide two parallel implementations of our algorithm: one using OpenMP, one MPI + OpenMP.
In our experiments against the state of the art, our algorithm (i) yields more accurate approximation results
for $\diag{\Lpinv}$, (ii) is much faster and more memory-efficient, and (iii) obtains good parallel speedups,
in particular in the distributed setting.
\end{abstract}

\newpage

\section{Introduction}
\label{sec:intro}
Massive graph data sets are abundant these days in numerous applications.
Extracting knowledge from these data thus requires fast algorithms.
One common matrix to represent a graph $G = (V,E)$ with $n$ vertices and $m$ edges
in algebraic algorithms is its Laplacian $\mat{L} = \mat{D}-\mat{A}$.
Here, $\mat{D}$ is the diagonal degree matrix with $\ment{D}{u}{u}$ being the (possibly weighted) degree of vertex
$u \in V$.
The matrix $\mat{A}$ is the (possibly weighted) adjacency matrix of $G$.
It is well-known that $\mat{L}$ does not have full rank and is thus not invertible.
Its Moore-Penrose pseudoinverse~\cite{DBLP:books/daglib/0086372} $\Lpinv$, in turn, has numerous
applications in physics and engineering~\cite{van2017pseudoinverse} as well as applied mathematics~\cite{DBLP:books/daglib/0086372}
and graph (resp.\ matrix) algorithms~\cite{DBLP:conf/www/0002PSYZ19}.

We are motivated by one particular class of applications:
\emph{electrical centrality measures}
for the analysis of small-world networks -- \ie graphs whose diameter is bounded by $\Oh(\log n)$.
Many important real-world networks (social, epidemiological,
information, biological, etc.) have the small-world feature~\cite{newman2018networks}.
Centrality measures, in turn, belong to the most widely used network analysis concepts
and indicate the importance of a vertex (or edge) in the network~\cite{DBLP:journals/im/BoldiV14}.
Numerous measures exist, some based on shortest paths, others consider paths of arbitrary
lengths. Electrical centrality measures fall into the latter category.
They exploit the perspective of graphs as electrical networks (see \eg~\cite{Lovasz1996}).
One well-known of such measures is
\emph{electrical closeness centrality}, a.\,k.\,a.\ \emph{current-flow closeness} or \emph{information centrality}~\cite{DBLP:conf/stacs/BrandesF05}, $c^{el}(\cdot)$.
It is the reciprocal of the average effective resistance $\effres{u}{\cdot}$ between $u$ and all other vertices:
\begin{equation}
  c^{el}(u) := \frac{n-1}{\sum_{v \in V \setminus \{ u \}} \effres{u}{v}}.
\end{equation}
In an electrical network corresponding to $G$, $\effres{u}{v}$ is the potential (voltage) difference
across terminals $u$ and $v$ when a unit current is applied between them~\cite{Ghosh:2008:MER:1350622.1350629}.
It can be computed by solving $\mat{L} \myvec{x} = \uvec{u} - \uvec{v}$ for $\myvec{x}$,
where $\uvec{z}$ is the canonical unit vector for vertex $z$. Then, $\effres{u}{v} = \vent{x}{u} - \vent{x}{v}$,
also see Section~\ref{sub:notation}.

Effective resistance also plays a major role in two other electrical measures we consider here,
\emph{normalized random-walk betweenness}~\cite{NarayanS18scaling} and \emph{Kirchhoff index centrality}~\cite{li2018kirchhoff}.
Also note that effective resistance is a graph metric with numerous other applications, well beyond its usage
in electrical centralities (cf. Refs.~\cite{DBLP:conf/innovations/AlevALG18,Ghosh:2008:MER:1350622.1350629}).
A straight\-forward way to compute electrical closeness (or the other two measures) would be to compute $\Lpinv$.
Without exploiting structure, this takes $\Oh(n^\omega)$ time,
where $\omega < 2.38$ is the exponent for fast matrix multiplication.
The standard tools in practice even require cubic time, cf.~\cite{RanjanZB14incremental}.
$\Lpinv$ is also in general a dense matrix (also for sparse $\mat{L}$).
Thus, full (pseudo)inversion is clearly limited to small inputs.

Conceptually similar to inversion would be to solve $\Theta(n)$ Laplacian linear systems.
In situations with lower accuracy demands, fewer linear systems suffice:
using the Johnson-Lindenstrauss transform (JLT) in connection with a fast
Laplacian solver such as Ref.~\cite{CohenKyng14},
one gets a relative approximation guarantee by solving $\Oh(\log n / \epsilon^2)$ systems~\cite{DBLP:journals/siamcomp/SpielmanS11}
in $\tilde{\Oh}(m \log^{1/2} n \log(1/\epsilon))$ time each,
where $\tilde{\Oh}(\cdot)$ hides a
$\Oh((\log \log n)^{3+\delta})$ factor for $\delta > 0$.

As pointed out previously~\cite{DBLP:journals/socnet/BozzoF13}, the (only) relevant part of $\Lpinv$
for computing electrical closeness is its diagonal (we will see that this is true for other
measures as well). Numerical methods for sampling-based approximation of the diagonal of implicitly
given matrices do exist~\cite{bekas2007est}. Yet, for our purpose, they solve
$\Oh(\log{n}/ \epsilon^2)$ Laplacian linear systems as well
to obtain an $\epsilon$-approximation with high probability, see Section~\ref{sub:related} for more details.

While this number of Laplacian linear systems can be
solved in parallel, their solution can still be time-consuming in practice, in part due to high constants hidden in the $\Oh$-notation.

\subparagraph{Contribution and Outline}
We propose a new algorithm for approximating $\diag{\Lpinv}$ of
a Laplacian matrix $\mat{L}$ that corresponds to weighted undirected graphs (Section~\ref{sec:algorithm}).
Our main technique is the approximation of effective resistances between
a pivot vertex $u \in V$ and all other vertices of $G$. It is based on sampling uniform (= random) spanning
trees (USTs). The resulting algorithm is highly parallel and (almost)
purely combinatorial -- it relies on the connection between Laplacian linear systems, effective resistances, and USTs.

For small-world graphs, our algorithm obtains an absolute $\pm\epsilon$-approximation guarantee
with high probability
in (sequential) time $\Oh(m \log^4 n \cdot \epsilon^{-2})$.
In particular, compared to using the fastest theoretical Laplacian solvers
in connection with JLT, 	our approach is off by only a polylogarithmic factor.
Probably more importantly, after some algorithm
engi\-nee\-ring (Section~\ref{sec:impl}), our algorithm performs
much better than the state of the art, already in our sequential experiments (Section~\ref{sec:experiments}):
(i) it is much faster and more memory-efficient,
(ii) it yields a maximum absolute error that is one order
of magnitude lower, and (iii) results in a more accurate complete
centrality ranking of elements of $\diag{\Lpinv}$.
Due to good parallel speedups, we can even compute a reasonably accurate diagonal of $\Lpinv$
on a small-scale cluster with 16 compute nodes in less than 8 minutes for a graph with $\approx 13.6$M vertices
and $\approx 334.6$M edges.
\ifthenelse{\boolean{confversion}}
{\textbf{Material omitted due to space constraints can be found in
  the appendix of the full version of this paper~\cite{angriman2020approximation}.}}
{\textbf{Material omitted from the main part can be found in
  the appendix.}}

\section{Preliminaries}
\label{sec:prelim}
%
\subsection{Problem Description and Notation}
\label{sub:notation}
We type vectors and matrices in bold font.
As input we consider simple, finite, connected undirected graphs $G = (V, E)$ with $n$ vertices,
$m$ edges, and non-negative edge weights $\myvec{w} \in \mathbb{R}_{\geq 0}^{m}$.
For the complexity analysis, we usually assume that the diameter of $G$ is $\Oh(\log n)$, but our algorithm would also work
correctly without this assumption.

\subparagraph{Graphs as electrical networks}
We interpret $G$ as an electrical network in which every edge $e \in E$ represents a resistor with resistance $1/\vent{w}{e}$.
In this context, it is customary to fix an arbitrary orientation $E^\pm$ of the edges in $E$ and to define
a unit $s$-$t$-current flow (also called electrical flow) in this network as a function
of the edges (written as vector) $\myvec{f} \in \mathbb{R}_{\geq 0}^{|E^{\pm}|}$.
Whenever possible, we use $\vent{f}{u,v}$ as shorthand notation for
$\vent{f}{\{u,v\}}$ or $\vent{f}{(u,v)}$. Note that $\vent{f}{e} = - \vent{f}{-e}$ for $e \notin E^\pm$.
This sign change in the $s$-$t$ current when the flow direction is
changed, is required to adhere to Kirchhoff's current law on flow conservation:
\begin{equation}
  \sum_{w \in \delta^+(v)} \vent{f}{v,w} - \sum_{u \in \delta^-(v)} \vent{f}{u,v} = \begin{cases}
    1 &  \text{if } v = s \\
    -1 & \text{if } v = t \\
    0 & \text{otherwise,}
  \end{cases}
\end{equation}
where $\delta^+(v)$ [$\delta^-(v)$] is the set of edges having $v$ as head [tail] in the orientation
we choose in $E^\pm$. Such a flow also adheres to Kirchhoff's voltage law (sum in cycle is zero
when considering flow directions) and Ohm's law (potential difference = resistance $\cdot$ current), cf.~\cite{10.5555/3086816,DBLP:books/daglib/0009415}.
The effective resistance between two vertices $u$ and $v$, $\effres{u}{v}$, is defined as
the potential
difference between $u$ and $v$ when a unit current is injected into $G$ at $u$ and extracted
at $v$, comp.~\cite[Ch.~IX]{DBLP:books/daglib/0009415}.
To compute $\effres{u}{v}$, let $\uvec{z}$ be the canonical unit vector for vertex $z$, \ie $\uvec{z}(z) = 1$ and
$\uvec{v} = 0$ for all vertices $v \neq z$. Then,
\begin{equation}
  \label{eq:eff-res}
  \effres{u}{v} = (\uvec{u} - \uvec{v})^T \Lpinv (\uvec{u} - \uvec{v}) = \ment{\Lpinv}{u}{u} - 2\ment{\Lpinv}{u}{v} + \ment{\Lpinv}{v}{v}
\end{equation}
or, equivalently, $\effres{u}{v} = \vent{x}{u} - \vent{x}{v}$, where $\myvec{x}$ is the solution vector of the Laplacian linear system $\mathbf{L} \myvec{x} = \uvec{u} - \uvec{v}$.
The Laplacian pseudoinverse, $\Lpinv$, can be expressed as $\Lpinv = (\mat{L} + \frac{1}{n} \mat{J})^{-1} - \frac{1}{n} \mat{J}$~\cite{van2017pseudoinverse},
where $\mat{J}$ is the $n \times n$-matrix with all entries being $1$.

Also note that the effective resistance between the endpoints of an
edge $e \in E$ equals the probability that $e$
is an edge in a uniform spanning tree (UST), \ie a spanning tree selected uniformly at random among
all spanning trees of $G$, cf.~\cite[Ch.~II]{DBLP:books/daglib/0009415}.

\subparagraph{Electrical Closeness}
The combinatorial counterpart of electrical closeness is based on shortest-path distances
$\operatorname{dist}(u,v)$ for vertices $u$ and $v$ in $G$:
$c^{c}(u) := \left(n-1\right) / \farnc{G}{u}$, where the denominator is the \emph{combinatorial farness} of $u$:
\begin{equation}
\label{eq:comb-farness}
\farnc{G}{u} := \sum_{v \in V \setminus \{u\}} \operatorname{dist}(u,v).
\end{equation}
Electrical farness $\farnel{G}{\cdot}$ is defined analogously to combinatorial farness -- shortest-path distances
in Eq.~(\ref{eq:comb-farness}) are replaced by effective resistances $\effres{u}{v}$.
Closeness centrality (both combinatorial and electrical) are not defined for disconnected graphs
due to infinite distances. We can get around this, however: a combinatorial generalization
for closeness called Lin's index (cf.~\cite{DBLP:journals/tkdd/BergaminiBCMM19}) can be adapted to the
electrical case, too. Thus, our assumption of $G$ being connected is no limitation.

\subparagraph{Normalized Random-Walk Betweenness}
Classical betweenness, based on shortest paths, is one of the most popular centrality measures.
The betweenness of vertices using random-walk routing
instead of shortest paths is given by the normalized random-walk
betweenness (NRWB)~\cite{NarayanS18scaling}.
This measure counts each random walk passing through a vertex only once.
By mapping the random walk problem to
current flowing in a network,
Ref.~\cite{NarayanS18scaling} obtains
a closed-form expression of NRWB and provides an analysis
of its scaling behavior as a function of $n$.
Then, the NRWB $c_b(\cdot)$ of a vertex $v$ is:
\begin{equation}\label{normalizerandomwalk}
  c_b(v) = \frac{1}{n}+\frac{1}{n-1} \sum_{t \neq v} \frac{\ment{M^{-1}}{t}{t} - \ment{M^{-1}}{t}{v}}{\ment{M^{-1}}{t}{t} + \ment{M^{-1}}{v}{v} -2\ment{M^{-1}}{t}{v}},
\end{equation}

where $ \mat{M} := \mat{L} + \mat{P} $, with $\mat{P}$ the projection
operator onto the zero eigenvector of the Laplacian $\mat{L}$ such that $\ment{P}{i}{j} = 1/n$. We show in Section~\ref{Sec:generalization} how to simplify this expression.

We also consider the Kirchhoff index and related centrality measures.
Their description can be found in
{Appendix~\ref{sub:app:Kirchhoff-descr}.}

\subsection{Related Work}
\label{sub:related}
\subparagraph*{Solving Laplacian Systems}
  A straightforward approach to compute electrical closeness and related centralities
  is to compute $\Lpinv$, by solving a number of Laplacian systems.
  Brandes and Fleischer~\cite{DBLP:conf/stacs/BrandesF05}
  computed electrical closeness from the solution of $n$ linear systems
  using conjugate gradient (CG) in $\Oh(m n \sqrt{\kappa})$ time,
  where $\kappa$ is the condition number of the appropriately preconditioned Laplacian matrix.\footnote{Brandes and Fleischer provide a rough estimate of $\kappa$ as $\Theta(n)$,
  leading to a total time of $\Oh(m n^{1.5})$.}
  Later, Spielman and Srivastava proposed an approximation algorithm for computing effective
  resistance distances~\cite{DBLP:journals/siamcomp/SpielmanS11}.
The main ingredients of the algorithm are a dimension reduction with
Johnson-Lindenstrauss~\cite{johnson1984extensions} and
the use of a fast Laplacian solver for
$\Oh(\log{n}/\epsilon^2)$ Laplacian systems.
The algorithm approximates effective resistance values for all edges
within a factor of $(1 \pm \epsilon)$
in $\Oh(I(n,m)\log{n} /\epsilon^2)$ time,
where $I(n, m)$ is the running time of the Laplacian solver,
assuming that the solution of the Laplacian systems is exact.
With an approximate Laplacian solution, the algorithm yields a
$(1+\epsilon)^2$-approximation.
Significant progress in the development of fast Laplacian solvers with theoretical
guarantees~\cite{DBLP:conf/focs/KoutisMP11, KoutisMPSiam14, DBLP:journals/KelnerOrecchia13, CohenKyng14,KyngPenSachdeva16}
has resulted in the currently best one running in $\Oh(m\log^{1/2}{n}\log(1/\epsilon))$ time
(up to polylogarithmic factors)~\cite{CohenKyng14}.
  Parallel algorithms for solving linear systems on the more general SDD matrices
  also exist in the literature~\cite{peng2013efficient,blelloch2011}.
  To date, the fastest algorithms for electrical closeness and spanning edge centrality
  extend the idea of Spielman and Srivastava~\cite{DBLP:conf/siamcsc/BergaminiWLM16, mavroforakis2015spanning,Hayashi2016EfficientAF}.
  (Similar ideas are used for centrality measures based on the Kirchhoff index,
  see Appendix~\ref{sub:app:Kirchhoff-rel-work}).
  Since theoretical Laplacian solvers
  rely on heavy graph-theoretic machinery such as low-stretch
  spanning trees, one rather uses
  multigrid solvers~\cite{livne2012lean,DBLP:conf/siamcsc/BergaminiWLM16,koutis2011combinatorial}
  in practice instead.

\subparagraph*{Diagonal Estimation}
  Recall from the introduction that the diagonal (or in case of the Kirchhoff index even only its sum, the trace --
  {see Appendix~\ref{sub:app:Kirchhoff-descr})}
  of $\Lpinv$ is enough to compute electrical closeness and related measures, also see Eq.~(\ref{eq:elec-farness}) below.
Algorithms that approximate the diagonal (or the trace) of matrices
that are only implicitly available often use iterative methods~\cite{Sidje2011},
sparse direct methods~\cite{DBLP:journals/siamsc/AmestoyDLR15,JACQUELIN201884},
Monte Carlo~\cite{Hutchinson90} or
deterministic probing techniques~\cite{Tang10aprobing, bekas2007est}.
A popular approach is the standard Monte-Carlo method for the trace of $\mat{A}$,
due to Hutchinson~\cite{Hutchinson90}. The idea is to estimate the trace of $\mat{A}$
by observing the action of $\mat{A}$ (in terms of matrix-vector products)
on a sufficiently large sample of random vectors $r_k$. In our case, this would require
to solve a large number of Laplacian linear systems with vectors $r_k$ as
right-hand sides.
Avron and Toledo~\cite{DBLP:journals/jacm/AvronT11} proved that the method takes
$\Oh(\log{n} / \epsilon^2)$ samples to achieve a maximum error $\epsilon$ with probability at least $1-\delta$.
The approach from Hutchinson~\cite{Hutchinson90} has been extended by Bekas
\etal~\cite{bekas2007est} for estimating the diagonal of $\mat{A}$.
Finally, Barthelm\'e \etal~\cite{barthelm2019estimating} have recently proposed
a combinatorial algorithm to approximate the trace (not the diagonal) of the
inverse of a matrix closely related to the Laplacian. Their algorithm can be
seen as a special case of our algorithm in the situation where a universal
vertex exists.

\subparagraph*{Normalized Random Walk Betweenness}
Along with the introduction of the measure,
Ref.~\cite{NarayanS18scaling}
provided numerical evaluations of  Eq.~(\ref{normalizerandomwalk})
on various graph models.
However, no algorithm to compute the measure without (pseudo)inverting $\mat{L}$ has been
proposed yet.

\section{Approximation Algorithm}
\label{sec:algorithm}
%
\subsection{Overview}
\label{sub:overview}
In order to compute the electrical closeness for all vertices in $V(G)$,
the main challenge is obviously to compute their electrical farness, $\farnel{G}{\cdot}$.
Recall from the introduction that the diagonal of $\Lpinv$ is sufficient to compute $\farnel{G}{\cdot}$
for all vertices simultaneously (comp.\ Ref.~\cite[Eq.~(15)]{DBLP:journals/socnet/BozzoF13}
with a slightly different definition of electrical closeness).
This follows from Eq.~(\ref{eq:eff-res}) and the fact that each row/column in $\Lpinv$ sums to $0$:
\vspace{-1ex}
\begin{align}
\label{eq:elec-farness}
\begin{split}
  \farnel{G}{u} & := \sum_{v \in V \setminus \{u\}} \effres{u}{v} = n \cdot \ment{\Lpinv}{u}{u} + \trace{\Lpinv} - 2 \sum_{v \in V} \ment{\Lpinv}{u}{v} = n \cdot \ment{\Lpinv}{u}{u} + \trace{\Lpinv},
\end{split}
\end{align}
since the trace $\trace{\cdot}$ is the sum over the diagonal entries.

We are interested in an approximation of $\diag{\Lpinv}$, since we do not necessarily need exact values for our
particular applications. To this end, we propose an approximation algorithm for which we give a rough overview first.
Our algorithm works best for small-world networks -- thus, we focus on this important input class.
Let $G$ be unweighted for now; we discuss the extension to weighted graphs
in Section~\ref{app:weighted-graphs}.
\begin{enumerate}
  \item Select\footnote{As we will see later on,
   one can improve the empirical running time
   when $u$ is not arbitrary, but chosen so as to have low eccentricity,
   \ie the length of its longest shortest path.
   The correctness and the asymptotic time complexity of the algorithm are not affected by the selection, though.} a pivot vertex $u \in V$ and solve the linear system $\mat{L}\myvec{x} = \uvec{u} - \frac{1}{n}\cdot \onesvec$,
    where $\onesvec = (1, \dots, 1)^T$.
        Out of all solutions $\myvec{x}$, we want the
      one such that $\myvec{x} \perp \onesvec$, since this unique normalized solution
      is equal to $\ment{\Lpinv}{:}{u}$, the column of
      $\Lpinv$ corresponding to $u$, also see Ref.~\cite[pp.~6-7]{van2017pseudoinverse}.
  \item Throughout the rest of this paper we denote $V' := V \setminus \{u\}$.
	As a direct consequence from Eq.~(\ref{eq:eff-res}), the diagonal entries $\ment{\Lpinv}{v}{v}$
    for all $v \in V'$ can be computed as:
    \begin{equation}
    \label{eq:diag-comp}
      \ment{\Lpinv}{v}{v} = \effres{u}{v} - \ment{\Lpinv}{u}{u} + 2 \ment{\Lpinv}{v}{u}.
    \end{equation}
  \item It remains to approximate these $n-1$ effective resistance values
      $\effres{\cdot}{\cdot}$. In order to do so, we employ Kirchhoff's
      theorem, which connects electrical flows with spanning
      trees~\cite[Ch.~II]{DBLP:books/daglib/0009415}. To this end, let $N$ be
      the total number of spanning trees of $G$;
      moreover, let $\numtrees{s}{a}{b}{t}$ be the number of spanning trees in which the unique path from $s$ to $t$
      traverses the edge $\{a,b\}$ \emph{in the direction from $a$ to $b$}.

	\medskip
    \begin{theorem}[Kirchhoff, comp.~\cite{DBLP:books/daglib/0009415}]
      \label{thm:Kirchhoff-N-values}
      Let $\vent{f}{a,b} := (\numtrees{s}{a}{b}{t} - \numtrees{s}{b}{a}{t}) / N$.
      Distribute current flows on the edges of $G$ by sending a current of size $\vent{f}{a,b}$
      from $a$ to $b$ for every edge $\{a,b\}$. Then there is a total current of size $1$
      from $s$ to $t$ satisfying Kirchhoff's laws.
    \end{theorem}
	\medskip

    As a result of Theorem~\ref{thm:Kirchhoff-N-values}, the effective
    resistance between $s$ and $t$ is the potential difference between $s$ and
    $t$ induced by the current-flow given by $\myvec{f}$. Vice versa, since the
    current flow is induced by potential differences (Ohm's law),
    one simply has to add the currents on a path from $s$ to $t$ to compute
    $\effres{s}{t}$ (see Eq.~(\ref{eq:N-values-in-path}) in
    Section~\ref{sub:tree-sampling}). Actually, as a proxy for the current
    flows, we use the (approximate) $N(\cdot)$-values mentioned in
    Theorem~\ref{thm:Kirchhoff-N-values}.
  \item It is impractical to compute exact values for $N$ (\eg by Kirchhoff's
      matrix-tree theorem~\cite{godsil2013algebraic}, which would require the
      determinant or all eigenvalues of $\Lpinv$) or $N(\cdot)$ for large
      graphs.
      Instead, we obtain approximations of the desired values via sampling:
      we sample a set of uniform spanning trees and determine the $N(\cdot)$-values by aggregation
      over all sampled trees. This approach provides a probabilistic absolute approximation guarantee.
\end{enumerate}

The pseudocode of the algorithm, in adjusted order,
is shown and discussed as Algorithm~\ref{alg:approx-diag-lpinv} in Appendix~\ref{app:approx-algo}.
Components and properties of Algorithm~\ref{alg:approx-diag-lpinv} are explained in the remainder
of Section~\ref{sec:algorithm}; reading Section~\ref{sub:tree-sampling} while/before studying the pseudocode is recommended.

Note that Steps 2-4 of the algorithm are entirely combinatorial.
Step~1 may or may not be combinatorial, depending on the Laplacian solver used.
Corresponding implementation choices are discussed in Section~\ref{sec:impl}.

\subsection{Effective Resistance Approximation by UST Sampling}
\label{sub:tree-sampling}
Extending and generalizing work by Hayashi \etal~\cite{Hayashi2016EfficientAF} on spanning edge centrality,
our main idea is to compute a sufficiently large sample of USTs and to aggregate the $N(\cdot)$-values
of the edges in these USTs.
Given $G$ and an electrical flow with source $u$ and sink $v$, recall that the effective resistance
between them is the potential difference $\vent{x}{u} - \vent{x}{v}$, where $\myvec{x}$ is the solution
vector in $\mat{L} \myvec{x} = \uvec{u} - \uvec{v}$. Since $\myvec{x}$ is a
potential and the electrical flow $\myvec{f}$ results from its difference, $\effres{u}{v}$ can be computed
given \emph{any} path $\langle u = v_0, v_1, \dots, v_{k-1}, v_k = v \rangle$ as:
\begingroup
\allowdisplaybreaks
\begin{align}
\begin{split}
  \effres{u}{v} & = \sum_{i=0}^{k-1} \vent{f}{v_i, v_{i+1}} \\
                & = 1/N \sum_{i=0}^{k-1} \left(\numtrees{u}{v_i}{v_{i+1}}{v} - \numtrees{u}{v_{i+1}}{v_i}{v}\right) \label{eq:N-values-in-path}.
\end{split}
\end{align}
\endgroup
Recall that the sign of the current flow changes if we traverse an edge against the flow direction.
This is reflected by the second summand in the sum of Eq.~(\ref{eq:N-values-in-path}).
Since we can choose any path from $u$ to $v$, for efficiency reasons we use \emph{one shortest} path
$P(v)$ per vertex $v \in V'$. We compute these paths with one breadth-first search
(BFS) with root $u$, resulting in a tree $B_u$ whose edges are considered as
implicitly directed from the root to the leaves.
For each vertex $v \in V'$, we maintain an estimate
$R[v]$ of $\effres{u}{v}$, which is initially set to $0$ for all $v$.
After all USTs have been processed, we divide all entries of $R$ by $\tau$, the number of sampled
trees, \ie $\tau$ takes the role of $N$ in Eq.~(\ref{eq:N-values-in-path}).

\subparagraph{Sampling USTs}
In total, we sample $\tau$ USTs, where $\tau$ depends on
the desired approximation guarantee and is determined later.
The choice of the UST algorithm depends on the input:
for general graphs, the algorithm by Schild~\cite{Schild18} with time complexity
$\Oh(m^{1+o(1)})$ is the fastest. Among others, it uses a sophisticated shortcutting
technique using fast Laplacian solvers to speed up the
classical Aldous-Broder~\cite{Aldous90,Broder89} algorithm.
For unweighted small-world graphs, however, Wilson's simple algorithm
using loop-erased random walks is in $\Oh(m \log n)$, as outlined in Appendix~\ref{app:Wilson}.
Thus, for our class of inputs, Wilson's algorithm is preferred.

\subparagraph{Data structures}
When computing the contribution of a UST $T$ to $N(\cdot)$,
we need to update for each edge $e = (a,b) \in E(T)$ its contribution
to $\numtrees{u}{a}{b}{v}$ and $\numtrees{u}{b}{a}{v}$, respectively -- for exactly every
vertex $v$ for which $(a,b)$ [or $(b,a)$] lies on $P(v)$.
Hence, the algorithm that aggregates the contribution of UST $T$
to $R$ will need to traverse $P(v)$ for each vertex $v \in V$.
To this end, we represent the BFS tree $B_u$ as an array of parent pointers for each vertex $v \in V$.
On the other hand, the tree $T$ can conveniently be represented by
storing a child and a sibling for each vertex $v \in V$.
Compared to other representations (such as adjacency lists),
this data structure can be constructed and traversed
with low constant overhead.

\subparagraph{Tree Aggregation}
After constructing a UST $T$, we process it to update the intermediate effective resistance
values $R[\cdot]$.
Note that we can discard $T$ afterwards and do not have to store the full sample,
which has a positive effect on the memory footprint of our algorithm.
The aggregation algorithm is shown as Algorithm~\ref{alg:aggregate} in
Appendix~\ref{app:approx-algo}.
Recall that we need to determine for each vertex $v$
and each edge $(a, b) \in P(v)$, whether $(a, b)$ or $(b, a)$ occurs
on the unique $u$-$v$ path in $T$.
To simplify this test, we root $T$ at $u$; hence, it is
	enough to check if $(a, b)$ [or $(b, a)$] appears above
	$v$ in $T$.
For general graphs, the test still incurs quadratic overhead in running time
(in particular, the number of vertex-edge pairs that
need to be considered is $\farnc{G}{u} = \sum_{v \in V'} |P(v)| = \mathcal{O}(n^2)$).
We remark that, perhaps surprisingly, a bottom-up traversal of $T$
does not improve on this, either; it is similarly difficult to determine all $R[v]$ that a given
$(a, b) \in E(T)$ contributes to (those $v$ form an arbitrary subset of
descendants of $b$ in $T$).
However, we can exploit the fact that on small-world networks,
the depth of $B_u$ can be controlled, \ie $\farnc{G}{u}$ is sub-quadratic.
To accelerate the test, we first compute a DFS data structure for $T$,
\ie we determine discovery and finish timestamps for all vertices in $V$, respectively.
For an arbitrary $v \in V$ and $(a, b) \in V \times V$, this
data structure allows us to answer
in constant time (i) whether either $(a, b)$ or $(b, a)$ is in $T$
and
(ii) if $(a, b) \in E(T)$, whether $v$ appears below $(a, b)$ in $T$.
Finally, we loop over all $v \in V$ and all $e = (a, b) \in P(v)$
and aggregate the contribution of $T$ to $\numtrees{u}{a}{b}{v}$.
To do so, we add [subtract] $1$ to [from] $R[v]$ if $e$ has the same [opposite] direction in $B_u$.
If $e$ is not in $B_u$, $R[v]$ does not change.

\subsection{Algorithm Analysis}
\label{sub:analysis}
The choice of the pivot $u$ has an effect
on the time complexity of our algorithm. The intuitive reason is that the
BFS tree $B_u$ should be shallow in order to have short paths to the root $u$,
which is achieved by a $u$ with small eccentricity.
The proofs of this subsection can be found in Appendices~\ref{sub:app-lemma-agg-time} and~\ref{sub:app-thm-time}.
Regarding aggregation, we obtain:
\begin{lemma}
\label{lem:aggregation-time}
Tree aggregation (Algorithm~\ref{alg:aggregate} in Appendix~\ref{app:approx-algo})
has time complexity $\Oh(\farnc{G}{u})$, which can be bounded by
$\Oh(n \cdot \ecc(u)) = \Oh(n \cdot \diam{G})$.
\end{lemma}

In high-diameter networks, the farness of $u$ can become quadratic (consider a path graph)
and thus problematic for large inputs.
In small-world graphs, however, we obtain $\Oh(n \log n)$ per aggregation.
We continue the analysis with the main algorithmic result.
\begin{theorem}
\label{thm:time-complexity}
Let $G$ be an undirected and unweighted graph with $n$ vertices, $m$ edges, diameter $\diam{G}$
and Laplacian $\mat{L} = \mat{L}(G)$.
Then, our diagonal approximation algorithm (Algorithm~\ref{alg:approx-diag-lpinv} in Appendix~\ref{app:approx-algo})
computes an approximation of $\diag{\Lpinv}$
with absolute error $\pm \epsilon$ with probability $1-\delta$ in time
$\Oh(m \cdot \ecc^3(u) \cdot \epsilon^{-2} \cdot \log(m/\delta))$.
For small-world graphs and with $\delta := 1/n$ to get high probability,
this yields a time complexity of $\Oh(m \log^4 n \cdot \epsilon^{-2}$).
\end{theorem}
Thus, for small-world networks, we have an approximation algorithm whose
running time is nearly-linear in $m$ (\ie linear up to a polylogarithmic
factor), quadratic in $1 / \epsilon$, and logarithmic in $1 / \delta$. By
choosing a \enquote{good} pivot $u$, it is often possible to improve the
running time of Algorithm~\ref{alg:approx-diag-lpinv} by a constant factor (\ie
without affecting the $\Oh$-notation). In particular, there are vertices $u$
with $\ecc(u)$ as low as $\frac 12 \diam{G}$.
A discussion on the algorithm's parallelization in the work-depth model can be
found in Appendix~\ref{sub:app:parallelism}.

\begin{remark}
If $G$ has constant diameter, Algorithm~\ref{alg:approx-diag-lpinv} has
time complexity $\Oh(m \log n \cdot \epsilon^{-2})$
to obtain an absolute $\epsilon$-approximation guarantee.
This is faster than the best JLT-based approximation
(which provides a relative guarantee instead).
\end{remark}

\subsection{Generalizations}
\label{Sec:generalization}
In this section we show how our algorithm can be adapted to work
for weighted graphs and for normalized random-walk betweenness as well.
The extensions to Kirchhoff-related indices are presented
in Appendix~\ref{sub:app:Kirchhoff-results}.

\subparagraph{Extension to Weighted graphs}
\label{app:weighted-graphs}
For an extension to weighted graphs, we need a weighted version
of Kirchhoff's theorem. To this end, the weight of a spanning tree $T$ is defined as the product of
the weights (= conductances) of its edges. Then, let $N^*$ be the sum of the weights of all spanning trees
of $G$; also, let $N^*_{s,t}(a, b)$ be the sum of the weights of all spanning trees in which
the unique path from $s$ to $t$ traverses the edge $\{a,b\}$ {in the direction from $a$ to $b$}.
\begin{theorem}[comp.~\cite{DBLP:books/daglib/0009415}, p.~46]
  \label{thm:Kirchhoff-N-values-weighted}
There is a distribution of currents satisfying Ohm's law and Kirchhoff's laws
in which a current of size $1$ enters at $s$ and leaves at $t$. The value of
the current on edge $\{a,b\}$ is given by $(N^*_{s,t}(a, b) - N^*_{s,t}(b, a))
/ N^*$.
\end{theorem}

Consequently, our sampling approach needs to estimate $N^*$ as well as the $N^*(\cdot)$-values.
It turns out that no major changes are necessary. Wilson's algorithm also yields a UST for weighted
graphs (if its random walk takes edge weights for transition probabilities into account)~\cite{Wilson:1996:GRS:237814.237880}.
Yet, the running time bound for Wilson needs to mention the graph volume,
$\operatorname{vol}(G)$, explicitly now: $\Oh(\ecc(u) \cdot \operatorname{vol}(G))$.
The weight of each sampled spanning tree can be accumulated during each run of Wilson.
It has to be integrated into Algorithm~\ref{alg:aggregate} by
adding [subtracting] the tree weight in Line~\ref{line:tree-add} [Line~\ref{line:tree-sub}]
instead of 1. For the division at the end, one has to replace $\tau$ by the total weight
of the sampled trees.
Finally, the tree $B_u$ remains a BFS tree. The eccentricity and farness of $u$ then still refer in the analysis
to their unweighted versions, respectively, as far as $B_u$ is concerned.

To conclude, the only important change regarding bounds happens in Theorem~\ref{thm:time-complexity}.
In the time complexity, $m$ is replaced by $\operatorname{vol}(G)$.

\subparagraph{Normalized Random-Walk Betweenness}
Ref.~\cite{NarayanS18scaling} proposes normalized random-walk betweenness
as a measure for the influence of a vertex
in the network, but the paper does not provide an
algorithm (beyond implicit (pseudo)inversion). We propose to compute normalized
random-walk betweenness with Algorithm~\ref{alg:approx-diag-lpinv}
and derive (proof
{in Appendix~\ref{app:nrwb}):}
\begin{lemma}
\label{lem:nrwb-derivation}
Normalized random-walk betweenness $c_b(v)$
(Eq.~(\ref{normalizerandomwalk}))
can be rewritten as:
\begin{align}
  \label{lem:new_norma_def}
  c_b(v) &= \frac{1}{n}+\frac{\trace{\Lpinv}}{(n-1)\farnel{G}{v}}.
\end{align}
\end{lemma}

Hence, since Algorithm~\ref{alg:approx-diag-lpinv} approximates the diagonal of $\Lpinv$
and both trace and electrical farness depend only on the diagonal,
the following proposition holds:

\begin{proposition}
  \label{prop:norma-betw}
  Let $G = (V, E)$ be a small-world graph as in Theorem~\ref{thm:time-complexity}. Then, Algorithm~\ref{alg:approx-diag-lpinv}
approximates with high probability $c_b(v)$ for all $v \in V$ with absolute error $\pm \epsilon$ in $\Oh(m \log^4 n \cdot \epsilon^{-2})$ time.
\end{proposition}

\section{Engineering Aspects and Parallelization}
\label{sec:impl}
Important engineering decisions concern the choice of the UST sampling
algorithm, decomposition of the input graph into biconnected components,
selection of the pivot $u$, and the linear solver used for the initial
linear system. For these aspects, we refer the reader to
{Appendix~\ref{sub:general-engineering}.}

Our implementation uses OpenMP for shared memory and MPI+OpenMP
for distributed memory. We assume that the entire graph fits
into main memory (even in the distributed case).
Hence, we can parallelize Algorithm~\ref{alg:approx-diag-lpinv} to a large extent by sampling
and aggregating multiple USTs in parallel.
In particular, we turn the main sampling loop into
a \textbf{parallel for} loop.
We also solve the initial Laplacian system using a
shared-memory parallel Conjugate Gradient (CG) solver
(see Appendix~\ref{sub:general-engineering}).
Note that we do not employ parallelism in the other steps
of the algorithm. In particular, the BFS to compute $B_u$
is executed sequentially. We also do not parallelize over the
loops in Algorithm~\ref{alg:aggregate} to avoid nested parallelism
with multiple invocations of Algorithm~\ref{alg:aggregate}.
We note that, in contrast to the theoretical
work-depth model, solving the initial Laplacian system
and performing the BFS are not the main bottlenecks in practice.
Instead, sampling and aggregating USTs together consume the majority
of CPU time (see
\ifthenelse{\boolean{confversion}}
{Figure~5 in Appendix~E.1 of the full version~\cite{angriman2020approximation}).}
{Figure~\ref{fig:breakdown} in Appendix~\ref{sub:app:par-scal}.}
More details regarding shared and distributed memory,
in particular load balancing for the distributed case,
are discussed
{in Appendix~\ref{sec:app:parallelism}.}

\section{Experiments}
\label{sec:experiments}
\subsection{Settings}
We conduct experiments to demonstrate the performance
of our approach compared to the state-of-the-art competitors.
Unless stated otherwise, we implemented all algorithms
in C++, using the NetworKit~\cite{DBLP:journals/netsci/StaudtSM16} graph APIs.
Our own algorithm is labelled \ust in the sections below.
All experiments were conducted on a cluster with 16 Linux machines, each
equipped with an Intel Xeon X7460 CPU (2 sockets, 12 cores each), and 192 GB of
RAM. To ensure reproducibility, all experiments were managed by the
SimexPal~\cite{angriman2019guidelines} software.
We executed our experiments on the graphs in
\ifthenelse{\boolean{confversion}}
{Tables 2, 3, 4, and 5}
{Tables~\ref{tab:networks-medium-gt}, \ref{tab:networks-medium},
\ref{tab:networks-large}, and \ref{tab:networks-cluster}}
\ifthenelse{\boolean{confversion}}
{(see Appendix~D.2 of the full version of this
paper~\cite{angriman2020approximation} for further details on input
graphs).}
{in Appendix~\ref{sub:app:instances}.}
All of them are unweighted and undirected. They have been downloaded from the
public KONECT~\cite{DBLP:conf/www/Kunegis13} repository and reduced to their
largest connected component.

\subparagraph{Quality Measures and Baseline}
To evaluate the diagonal approximation quality, we measure the maximum absolute
error ($\max_i \Lpinv_{ii} - \widetilde{\Lpinv_{ii}}$) on each instance, and we
take both the maximum and the arithmetic mean over all the instances.
Since for some applications~\cite{newman2018networks,okamoto2008ranking} a correct ranking of the
entries is more relevant than their scores, in our experimental
analysis we compare complete rankings of the elements of $\widetilde{\Lpinv}$.
Note that the lowest entries of $\Lpinv$ (corresponding to the vertices with highest
electrical closeness) are distributed on a significantly narrow interval.
Hence, to achieve an accurate electrical closeness ranking of the top $k$ vertices,
one would need to solve the problem with very high accuracy.
For this reason, all approximation algorithms we consider do not yield a precise
top-$k$ ranking, so that we (mostly) consider the complete ranking.

Using \texttt{pinv} in NumPy or Matlab as a baseline would be too
expensive in terms of time (cubic) and space (quadratic) on large graphs
(see
\ifthenelse{\boolean{confversion}}
{Appendix~D.2 of the full version~\cite{angriman2020approximation}}
{Appendix~\ref{sub:app:instances}}).
  Thus, as quality baseline we employ the LAMG
  solver~\cite{livne2012lean} (see also next paragraph)
  as implemented within \nwk~\cite{DBLP:conf/siamcsc/BergaminiWLM16} in
  our experiments (with $\lamgTol$ tolerance).
The results in
\ifthenelse{\boolean{confversion}}
{Table~6 in Appendix~E.4}
{Table~\ref{tab:lamg_pinv_table} in Appendix~\ref{app:baseline}}
indicate that
the diagonal obtained this way is sufficiently accurate.

\subparagraph{Competitors in Practice}
In practice, the fastest way to compute electrical closeness so far
is to combine a dimension reduction via the Johnson-Lindenstrauss
lemma~\cite{johnson1984extensions} (JLT) with a numerical solver.
In this context, Algebraic MultiGrid (AMG) solvers exhibit
better empirical running time than fast Laplacian
solvers with a worst-case guarantee~\cite{mavroforakis2015spanning}.
For our experiments we use JLT combined with LAMG~\cite{livne2012lean} (named \jltLamg);
the latter is an AMG-type solver for complex networks.
We also compare against a Julia implementation of JLT together with the fast
Laplacian solver proposed by Kyng \etal~\cite{kyng16}, for which a Julia implementation
is already available in the package \tool{Laplacians.jl}.\footnote{ \url{https://github.com/danspielman/Laplacians.jl}}
This solver generates a sparse approximate Cholesky
decomposition for Laplacian matrices with provable approximation guarantees in
$\Oh(m\log^3{n}\log(1/\epsilon))$ time; it is based purely on random
sampling (and does not make use of graph-theoretic concepts such
as low-stretch spanning trees or sparsifiers).
We refer to the above implementation as \jltKyng
throughout the experiments. For both \jltLamg and \jltKyng, we try different input error bounds
(they correspond to the respective numbers next to the method names in Figure~\ref{fig:quality-rt}).
This is a relative error, since these algorithms use numerical approaches with a
relative error guarantee, instead of an absolute one
(see
\ifthenelse{\boolean{confversion}}
{Appendix~E of the full version~\cite{angriman2020approximation}}
{Appendix~\ref{app:additional_experiments}}
for results in terms of different quality measures).

Finally, we compare against the diagonal estimators due to Bekas \etal~\cite{bekas2007est},
one based on random vectors and one based on Hadamard rows.
To solve the resulting Laplacian systems, we use LAMG in both cases.
In our experiments, the algorithms are referred to as \bekas and \bekasH, respectively.
Excluded competitors are discussed in
\ifthenelse{\boolean{confversion}}
{Appendix~D.1 (full version~\cite{angriman2020approximation}).}
{Appendix~\ref{sub:app:excluded}.}

\subsection{Running Time and Quality}
\label{exp:quality-results}
\begin{figure}
\centering
\begin{subfigure}[t]{\textwidth}
\centering
\includegraphics{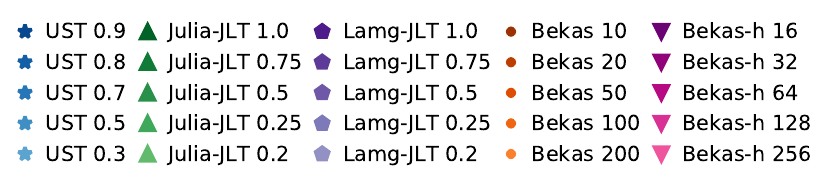}
\end{subfigure}

\begin{subfigure}[t]{.33\textwidth}
\includegraphics{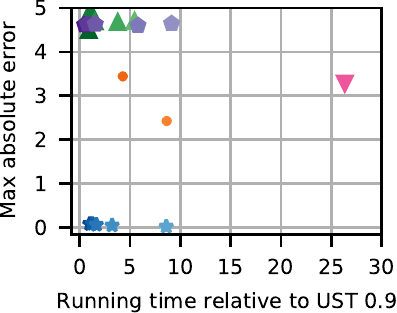}
\caption{Maximum of the maximum absolute errors.}
\label{fig:max-abs-err}
\end{subfigure}\hfill
\begin{subfigure}[t]{.33\textwidth}
\includegraphics{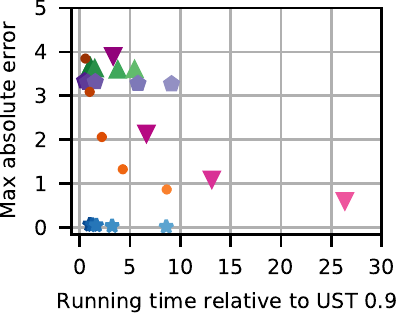}
\caption{Arithmetic mean of the maximum absolute error.}
\label{fig:avg-abs-err}
\end{subfigure}\hfill
\begin{subfigure}[t]{.33\textwidth}
\centering
\includegraphics{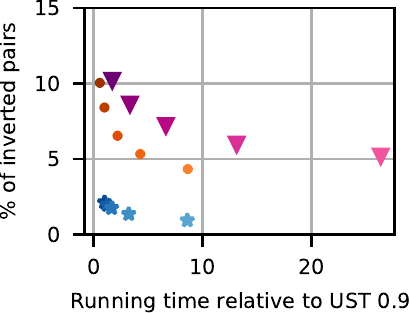}
\caption{Geometric mean of the percentage of inverted pairs in the full ranking of $\diag{\Lpinv}$.
}
\label{fig:full-ranking}
\end{subfigure}

\caption{Quality results over the instances of
\ifthenelse{\boolean{confversion}}{Table 2 (full version~\cite{angriman2020approximation})}{
Table~\ref{tab:networks-medium-gt}}.
All runs are sequential. }
\label{fig:quality-rt}
\end{figure}

Figure~\ref{fig:max-abs-err} shows that, in terms of maximum absolute error,
every configuration of \ust achieves results with higher quality than the
competitors. Even when setting $\epsilon = 0.9$, \ust yields a maximum absolute
error of \ustMaxAbsErr, and it is \ustBekasSpeedup faster than \bekas with 200
random vectors, which instead achieves a maximum absolute error of
\bekasMaxAbsErr.
Furthermore, the running time of \ust does not increase substantially for lower
values of $\epsilon$, and its quality does not deteriorate quickly for higher
values of $\epsilon$. Regarding the average of the maximum absolute error,
Figure~\ref{fig:avg-abs-err} shows that, among the competitors, \bekasH with
256 Hadamard rows achieves the best precision. However, \ust yields an average
error of \ustAvgAbsErr while also being \ustBekasHSpeedup faster than \bekasH,
which yields an average error of \bekasAvgAbsErr.
Note also that the number next to each method in Figure~\ref{fig:quality-rt} corresponds to
different values of absolute (for \ust) or relative (for \jltLamg, \jltKyng) error bounds, and different
numbers of samples (for \bekas, \bekasH). For \bekasH the number of samples needs to be a
multiple of four due to the dimension of Hadamard matrices.

In Figure~\ref{fig:full-ranking} we report the percentage of inverted pairs in
the full ranking of $\widetilde{\Lpinv}$.
Note that JLT-based approaches are not
  depicted in this plot, because they yield $>15\%$ of rank inversions.
Among the competitors, \bekas achieves the best time-accuracy trade-off. However,
when using 200 random vectors, it yields \bekasRanking inversions while also being
\ustBekasSpeedup slower than \ust with $\epsilon = 0.9$, which yields
\ustRanking inversions only.

For validation purposes, we also measure how well the considered
algorithms compute the \emph{set} (not the ranking) of top-$k$ vertices, \ie those with highest electrical
closeness centrality, with $k\in \{10, 100\}$. For each algorithm we only consider
the parameter settings that yields the highest accuracy.
JLT-based approaches appear to be
very accurate for this purpose, as their top-$k$ sets achieve a Jaccard
index of $1.0$.
As expected (due to its absolute error guarantee),
\ust performs slightly worse:
on average, it obtains \numprint{0.95} for $k = 10$ and \numprint{0.98} for $k = 100$, which still
shows a high overlap with the ground truth.

The memory consumption is shown and discussed in more detail
in \ifthenelse{\boolean{confversion}}{Appendix~E.3 (full version~\cite{angriman2020approximation})}{Appendix~\ref{sec:memory}}.
In summary, as our algorithm can discard each UST after its aggregation, it is
rather space-efficient and requires less memory than the competitors.

\subsection{Parallel Scalability}
\label{exp:parallel-scalability}
The log-log plot in \ifthenelse{\boolean{confversion}}{Figure 4a
(Appendix~E.1, full version~\cite{angriman2020approximation})}{
Figure~\ref{fig:shared-mem-scalability} (Appendix~\ref{sub:app:par-scal})}
shows that on shared-memory \ust achieves a moderate parallel scalability \wrt
the number of cores; on 24 cores in particular it is \shmemspeedup faster than
on a single core. Even though the number of USTs to be sampled can be evenly
divided among the available cores, we do not see a nearly-linear scalability:
on multiple cores the memory performance of our NUMA system becomes a
bottleneck. Therefore, the time to sample a UST increases and using more cores
yields diminishing returns. Limited memory bandwidth is a known issue affecting
algorithms based on graph traversals in
general~\cite{bader2005architectural,lumsdaine2007challenges}. Finally, we
compare the parallel performance of \ust indirectly with the parallel
performance of our competitors. More precisely, assuming a perfect parallel
scalability for our competitors \bekas and \bekasH on 24 cores, \ust would
yield results \numprint{4.1} and \numprint{12.6} times faster, respectively,
even with this strong assumption for the competition's benefit.

\ust scales better in a distributed setting.
In this case, the scalability is affected mainly by its non-parallel parts
and by synchronization latencies.
The log-log plot in
\ifthenelse{\boolean{confversion}}{Figure~4b}{Figure~\ref{fig:distr-mem-scalability}}
shows that on up to 4 compute nodes the scalability is almost linear, while on 16
compute nodes \ust achieves a \distrspeedup speedup \wrt a single compute node.

\ifthenelse{\boolean{confversion}}
{Figure~5 (Appendix~E.1, full version~\cite{angriman2020approximation})}
{Figure~\ref{fig:breakdown} in Appendix~\ref{sub:app:par-scal}}
shows the fraction of time that \ust
spends on different tasks depending on the number of cores. We aggregated over
\quot{Sequential Init.} the time spent on memory allocation, pivot selection,
solving the linear system, the computation of the biconnected components,
and on computing the tree $B_u$.
In all configurations, \ust spends the majority of the time in sampling, computing the DFS data structures and
aggregating USTs. The total time spent on aggregation corresponds to
\quot{UST aggregation} and \quot{DFS} in
\ifthenelse{\boolean{confversion}}{Figure~5}{Figure~\ref{fig:breakdown}},
indicating that computing the DFS data structures is the most expensive part of
the aggregation.
Together, sampling time and total aggregation time account for \ustFracOneCore and
  \ustFracAllCores of the total running time on 1 core and 24 cores, respectively.
  On average, sampling takes \avgSampling of this time, while total aggregation
  takes \avgAggregate. Since sampling a UST is on average \aggOverSampl more
  expensive than computing the DFS timestamps and aggregation, faster sampling techniques would significantly
  improve the performance of our algorithm.

\begin{figure}
\centering
\begin{subfigure}[t]{.5\textwidth}
\centering
\includegraphics{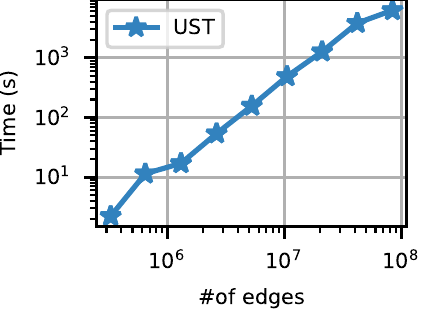}
\caption{Running time of \ust \wrt \#of edges.}
\label{fig:hyp-time}
\end{subfigure}\hfill
\begin{subfigure}[t]{.5\textwidth}
\centering
\includegraphics{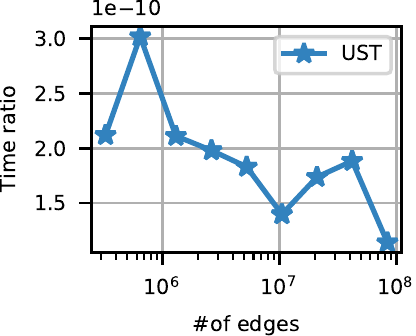}
\caption{Ratio of running time of \ust \wrt its theoretical running time
(see Theorem~\ref{thm:time-complexity}).}
\label{fig:hyp-ratio}
\end{subfigure}
\caption{Scalability of \ust on random hyperbolic graphs ($\epsilon = 0.3$, $1\times24$ cores).}
\label{fig:hyp}
\end{figure}

\subsection{Scalability to Large Networks}
\subparagraph{Results on Synthetic Networks}

The log-log plots in Figure~\ref{fig:hyp-time} show the average running
time of \ust ($1 \times 24$ cores) on networks generated with the random
hyperbolic generator from von Looz
\etal~\cite{DBLP:conf/hpec/LoozOLM16}.\footnote{The random hyperbolic generator
generates networks with a heavy-tailed degree distribution. We set the
average degree to $20$ and the exponent of the power-law distribution to $3$.}
For each network size, we take the arithmetic mean of the running times measured
for five different randomly generated networks.
Our algorithm requires \maxTimeHyp minutes for the largest inputs
(with up to \maxEdgesHyp million edges).
Interestingly, Figure~\ref{fig:hyp-ratio} shows that the algorithm scales slightly
better than our theoretical bound predicts.
In \ifthenelse{\boolean{confversion}}{Figure~6 (Appendix~E.2, full
version~\cite{angriman2020approximation})}{ Figure~\ref{fig:rmat}
(Appendix~\ref{sub:rmat})} we present results on an additional graph class,
namely R-MAT graphs.
On these instances, the algorithm exhibits a similar running time
behavior; however, the comparison to the theoretical bound
is less conclusive.

\subparagraph{Results on Large Real-World Networks}
\begin{table}[tb]
\small
\setlength{\tabcolsep}{\instTabColSep}
\centering
\caption{Running time of \ust on large real-world networks ($16\times24$ cores).}
\label{tab:cluster-only-time}
\begin{tabular}{lrrrr}
\multirow{2}{*}{Network} & \multirow{2}{*}{$|V|$} & \multirow{2}{*}{$|E|$} & Time (s)  & Time (s) \\
& & & $\epsilon = 0.3$ & $\epsilon = 0.9$ \\
\midrule
petster-carnivore & \numprint{601213} & \numprint{15661775} & \numprint{16.8} & \numprint{4.8}\\
soc-pokec-relationships & \numprint{1632803} & \numprint{22301964} & \numprint{55.5} & \numprint{9.5}\\
soc-LiveJournal1 & \numprint{4843953} & \numprint{42845684} & \numprint{277.0} & \numprint{75.5}\\
livejournal-links & \numprint{5189808} & \numprint{48687945} & \numprint{458.4} & \numprint{80.6}\\
orkut-links & \numprint{3072441} & \numprint{117184899} & \numprint{71.8} & \numprint{19.9}\\
wikipedia\_link\_en & \numprint{13591759} & \numprint{334590793} & \numprint{429.9} & \numprint{88.3}\\
\midrule
\end{tabular}

\end{table}

In Table~\ref{tab:cluster-only-time} we report the performance of \ust
in a distributed setting ($16\times 24$ cores) on large \emph{real-world} networks.
With $\epsilon = 0.3$ and $\epsilon = 0.9$, \ust always runs in less than
$8$ minutes and $\numprint{1.5}$ minutes, respectively.

\section{Conclusions}
\label{sec:conclusions}
We have proposed a new parallel algorithm for approximating $\diag{\Lpinv}$
of Laplacian matrices $\mat{L}$ corresponding to small-world networks.
Compared to the main competitors, our algorithm is about one order of magnitude faster,
it yields results with higher quality in terms of absolute error and
ranking of $\diag{\Lpinv}$, and it requires less memory.
The gap between the theoretical bounds and the much better empirical error yielded
    by our algorithm suggests that tighter bounds on the number of samples are a
    promising direction for future work.
So is an improvement of the running time for high-diameter graphs, both in theory
and practice.

\paragraph*{Acknowledgements}
This work is partially supported by German Research Foundation (DFG) grant ME 3619/3-2
within Priority Programme 1736 \textit{Algorithms for Big Data} and by DFG grant
ME 3619/4-1 (\textit{Accelerating Matrix Computations for Mining Large Dynamic Complex Networks}).
We thank our colleague Fabian Brandt-Tumescheit for his technical support for the experiments.

\newpage

\bibliography{references}

\begin{thebibliography}{10}

\bibitem{Aldous90}
David~J. Aldous.
\newblock The random walk construction of uniform spanning trees and uniform
  labelled trees.
\newblock {\em SIAM J. Discret. Math.}, 3(4):450–465, November 1990.
\newblock \href {https://doi.org/10.1137/0403039} {\path{doi:10.1137/0403039}}.

\bibitem{DBLP:conf/innovations/AlevALG18}
Vedat~Levi Alev, Nima Anari, Lap~Chi Lau, and Shayan~Oveis Gharan.
\newblock Graph clustering using effective resistance.
\newblock In Anna~R. Karlin, editor, {\em 9th Innovations in Theoretical
  Computer Science Conference, {ITCS} 2018, January 11-14, 2018, Cambridge, MA,
  {USA}}, volume~94 of {\em LIPIcs}, pages 41:1--41:16, Cambridge,
  Massachusetts, USA, 2018. Schloss Dagstuhl - Leibniz-Zentrum fuer Informatik.
\newblock \href {https://doi.org/10.4230/LIPIcs.ITCS.2018.41}
  {\path{doi:10.4230/LIPIcs.ITCS.2018.41}}.

\bibitem{DBLP:journals/siamsc/AmestoyDLR15}
Patrick Amestoy, Iain~S. Duff, Jean{-}Yves L'Excellent, and
  Fran{\c{c}}ois{-}Henry Rouet.
\newblock Parallel computation of entries of a\({}^{\mbox{-1}}\).
\newblock {\em {SIAM} J. Scientific Computing}, 37(2):C268--C284, 2015.
\newblock \href {https://doi.org/10.1137/120902616}
  {\path{doi:10.1137/120902616}}.

\bibitem{angriman2020approximation}
Eugenio Angriman, Maria Predari, Alexander van~der Grinten, and Henning
  Meyerhenke.
\newblock Approximation of the diagonal of a laplacian's pseudoinverse for
  complex network analysis, 2020.
\newblock \href {http://arxiv.org/abs/2006.13679} {\path{arXiv:2006.13679}}.

\bibitem{angriman2019guidelines}
Eugenio Angriman, Alexander van~der Grinten, Moritz von Looz, Henning
  Meyerhenke, Martin N{\"o}llenburg, Maria Predari, and Charilaos Tzovas.
\newblock Guidelines for experimental algorithmics: A case study in network
  analysis.
\newblock {\em Algorithms}, 12(7):127, 2019.

\bibitem{DBLP:journals/jacm/AvronT11}
Haim Avron and Sivan Toledo.
\newblock Randomized algorithms for estimating the trace of an implicit
  symmetric positive semi-definite matrix.
\newblock {\em J. {ACM}}, 58(2):8:1--8:34, 2011.
\newblock \href {https://doi.org/10.1145/1944345.1944349}
  {\path{doi:10.1145/1944345.1944349}}.

\bibitem{bader2005architectural}
David~A Bader, Guojing Cong, and John Feo.
\newblock On the architectural requirements for efficient execution of graph
  algorithms.
\newblock In {\em 2005 International Conference on Parallel Processing
  (ICPP'05)}, pages 547--556, Oslo, Norway, 2005. IEEE, IEEE.

\bibitem{barthelm2019estimating}
Simon Barthelmé, Nicolas Tremblay, Alexandre Gaudillière, Luca Avena, and
  Pierre-Olivier Amblard.
\newblock Estimating the inverse trace using random forests on graphs, 2019.
\newblock \href {http://arxiv.org/abs/1905.02086} {\path{arXiv:1905.02086}}.

\bibitem{bekas2007est}
C.~Bekas, E.~Kokiopoulou, and Y.~Saad.
\newblock An estimator for the diagonal of a matrix.
\newblock {\em Appl. Numer. Math.}, 57(11-12):1214--1229, November 2007.
\newblock URL: \url{http://dx.doi.org/10.1016/j.apnum.2007.01.003}, \href
  {https://doi.org/10.1016/j.apnum.2007.01.003}
  {\path{doi:10.1016/j.apnum.2007.01.003}}.

\bibitem{DBLP:journals/tkdd/BergaminiBCMM19}
Elisabetta Bergamini, Michele Borassi, Pierluigi Crescenzi, Andrea Marino, and
  Henning Meyerhenke.
\newblock Computing top-\emph{k} closeness centrality faster in unweighted
  graphs.
\newblock {\em {TKDD}}, 13(5):53:1--53:40, 2019.
\newblock \href {https://doi.org/10.1145/3344719} {\path{doi:10.1145/3344719}}.

\bibitem{DBLP:conf/siamcsc/BergaminiWLM16}
Elisabetta Bergamini, Michael Wegner, Dimitar Lukarski, and Henning Meyerhenke.
\newblock Estimating current-flow closeness centrality with a multigrid
  laplacian solver.
\newblock In {\em Proc. 7th {SIAM} Workshop on Combinatorial Scientific
  Computing, {CSC} 2016}, pages 1--12. {SIAM}, 2016.
\newblock \href {https://doi.org/10.1137/1.9781611974690.ch1}
  {\path{doi:10.1137/1.9781611974690.ch1}}.

\bibitem{blelloch2011}
Guy~E. Blelloch, Anupam Gupta, Ioannis Koutis, Gary~L. Miller, Richard Peng,
  and Kanat Tangwongsan.
\newblock Near linear-work parallel sdd solvers, low-diameter decomposition,
  and low-stretch subgraphs, 2011.
\newblock \href {https://doi.org/10.1145/1989493.1989496}
  {\path{doi:10.1145/1989493.1989496}}.

\bibitem{DBLP:journals/im/BoldiV14}
Paolo Boldi and Sebastiano Vigna.
\newblock Axioms for centrality.
\newblock {\em Internet Mathematics}, 10(3-4):222--262, 2014.
\newblock \href {https://doi.org/10.1080/15427951.2013.865686}
  {\path{doi:10.1080/15427951.2013.865686}}.

\bibitem{DBLP:books/daglib/0009415}
B{\'{e}}la Bollob{\'{a}}s.
\newblock {\em Modern Graph Theory}, volume 184 of {\em Graduate Texts in
  Mathematics}.
\newblock Springer, 2002.
\newblock \href {https://doi.org/10.1007/978-1-4612-0619-4}
  {\path{doi:10.1007/978-1-4612-0619-4}}.

\bibitem{bozzo2012approximations}
Enrico Bozzo and Massimo Franceschet.
\newblock Approximations of the generalized inverse of the graph laplacian
  matrix.
\newblock {\em Internet mathematics}, 8(4):456--481, 2012.

\bibitem{DBLP:journals/socnet/BozzoF13}
Enrico Bozzo and Massimo Franceschet.
\newblock Resistance distance, closeness, and betweenness.
\newblock {\em Social Networks}, 35(3):460--469, 2013.
\newblock \href {https://doi.org/10.1016/j.socnet.2013.05.003}
  {\path{doi:10.1016/j.socnet.2013.05.003}}.

\bibitem{DBLP:conf/stacs/BrandesF05}
Ulrik Brandes and Daniel Fleischer.
\newblock Centrality measures based on current flow.
\newblock In {\em Proceedings of the 22nd Annual Symposium on Theoretical
  Aspects of Computer Science, {STACS} 2005}, volume 3404 of {\em LNCS}, pages
  533--544. Springer, 2005.
\newblock URL: \url{http://dx.doi.org/10.1007/978-3-540-31856-9_44}, \href
  {https://doi.org/10.1007/978-3-540-31856-9_44}
  {\path{doi:10.1007/978-3-540-31856-9_44}}.

\bibitem{Broder89}
A.~Broder.
\newblock Generating random spanning trees.
\newblock In {\em Proceedings of the 30th Annual Symposium on Foundations of
  Computer Science}, SFCS ’89, page 442–447, USA, 1989. IEEE Computer
  Society.
\newblock \href {https://doi.org/10.1109/SFCS.1989.63516}
  {\path{doi:10.1109/SFCS.1989.63516}}.

\bibitem{chakrabarti2004r}
Deepayan Chakrabarti, Yiping Zhan, and Christos Faloutsos.
\newblock R-mat: A recursive model for graph mining.
\newblock In {\em Proceedings of the 2004 SIAM International Conference on Data
  Mining}, pages 442--446. SIAM, SIAM, 2004.

\bibitem{chandra1996electrical}
Ashok~K Chandra, Prabhakar Raghavan, Walter~L Ruzzo, Roman Smolensky, and
  Prasoon Tiwari.
\newblock The electrical resistance of a graph captures its commute and cover
  times.
\newblock {\em Computational Complexity}, 6(4):312--340, 1996.

\bibitem{CohenKyng14}
Michael~B. Cohen, Rasmus Kyng, Gary~L. Miller, Jakub~W. Pachocki, Richard Peng,
  Anup~B. Rao, and Shen~Chen Xu.
\newblock Solving sdd linear systems in nearly mlog1/2n time.
\newblock In {\em Proceedings of the Forty-Sixth Annual ACM Symposium on Theory
  of Computing}, STOC ’14, page 343–352, New York, NY, USA, 2014.
  Association for Computing Machinery.
\newblock \href {https://doi.org/10.1145/2591796.2591833}
  {\path{doi:10.1145/2591796.2591833}}.

\bibitem{Ellens2011}
Wendy Ellens, Flora Spieksma, P.~Mieghem, A.~Jamakovic, and Robert Kooij.
\newblock Effective graph resistance.
\newblock {\em Linear Algebra and its Applications}, 435:2491--2506, 11 2011.
\newblock \href {https://doi.org/10.1016/j.laa.2011.02.024}
  {\path{doi:10.1016/j.laa.2011.02.024}}.

\bibitem{ericson2013effective}
Josh Ericson, Pietro Poggi-Corradini, and Hainan Zhang.
\newblock Effective resistance on graphs and the epidemic quasimetric.
\newblock {\em Involve, a Journal of Mathematics}, 7(1):97--124, 2013.

\bibitem{Hutchinson90}
Hutchinson~M. F.
\newblock A stochastic estimator of the trace of the influence matrix for
  laplacian smoothing splines.
\newblock {\em J. Commun. Statist. Simula.}, 19(2):433--450, 1990.

\bibitem{Ghosh:2008:MER:1350622.1350629}
Arpita Ghosh, Stephen Boyd, and Amin Saberi.
\newblock Minimizing effective resistance of a graph.
\newblock {\em SIAM Rev.}, 50(1):37--66, February 2008.
\newblock URL: \url{http://dx.doi.org/10.1137/050645452}, \href
  {https://doi.org/10.1137/050645452} {\path{doi:10.1137/050645452}}.

\bibitem{godsil2013algebraic}
Chris Godsil and Gordon~F Royle.
\newblock {\em Algebraic graph theory}, volume 207.
\newblock Springer Science \& Business Media, 2013.

\bibitem{DBLP:books/daglib/0086372}
Gene~H. Golub and Charles F.~Van Loan.
\newblock {\em Matrix computations}.
\newblock Johns Hopkins University Press, 1996.

\bibitem{eigenweb}
Ga\"{e}l Guennebaud, Beno\^{i}t Jacob, et~al.
\newblock Eigen v3.
\newblock http://eigen.tuxfamily.org, 2010.

\bibitem{Hayashi2016EfficientAF}
Takanori Hayashi, Takuya Akiba, and Yuichi Yoshida.
\newblock Efficient algorithms for spanning tree centrality.
\newblock In {\em IJCAI}, pages 3733--3739. IJCAI, 2016.

\bibitem{JACQUELIN201884}
Mathias Jacquelin, Lin Lin, and Chao Yang.
\newblock Pselinv – a distributed memory parallel algorithm for selected
  inversion: The non-symmetric case.
\newblock {\em Parallel Computing}, 74:84 -- 98, 2018.
\newblock Parallel Matrix Algorithms and Applications (PMAA'16).
\newblock URL:
  \url{http://www.sciencedirect.com/science/article/pii/S0167819117301941},
  \href {https://doi.org/https://doi.org/10.1016/j.parco.2017.11.009}
  {\path{doi:https://doi.org/10.1016/j.parco.2017.11.009}}.

\bibitem{johnson1984extensions}
William~B Johnson and Joram Lindenstrauss.
\newblock Extensions of lipschitz mappings into a hilbert space.
\newblock {\em Contemporary mathematics}, 26(189-206):1, 1984.

\bibitem{DBLP:journals/KelnerOrecchia13}
Jonathan~A Kelner, Lorenzo Orecchia, Aaron Sidford, and Zeyuan~Allen Zhu.
\newblock A simple, combinatorial algorithm for solving sdd systems in
  nearly-linear time.
\newblock In {\em Proceedings of the forty-fifth annual ACM symposium on Theory
  of computing}, pages 911--920. ACM, 2013.

\bibitem{Klein93}
Douglas Klein and Milan Randic.
\newblock Resistance distance.
\newblock {\em Journal of Mathematical Chemistry}, 12:81--95, 12 1993.
\newblock \href {https://doi.org/10.1007/BF01164627}
  {\path{doi:10.1007/BF01164627}}.

\bibitem{DBLP:conf/focs/KoutisMP11}
Ioannis Koutis, Gary~L. Miller, and Richard Peng.
\newblock A nearly-m log n time solver for {SDD} linear systems.
\newblock In Rafail Ostrovsky, editor, {\em {IEEE} 52nd Annual Symposium on
  Foundations of Computer Science, {FOCS} 2011, Palm Springs, CA, USA, October
  22-25, 2011}, pages 590--598. {IEEE} Computer Society, 2011.
\newblock \href {https://doi.org/10.1109/FOCS.2011.85}
  {\path{doi:10.1109/FOCS.2011.85}}.

\bibitem{KoutisMPSiam14}
Ioannis Koutis, Gary~L. Miller, and Richard Peng.
\newblock Approaching optimality for solving {SDD} linear systems.
\newblock {\em {SIAM} J. Comput.}, 43(1):337--354, 2014.
\newblock URL: \url{http://dx.doi.org/10.1137/110845914}, \href
  {https://doi.org/10.1137/110845914} {\path{doi:10.1137/110845914}}.

\bibitem{koutis2011combinatorial}
Ioannis Koutis, Gary~L. Miller, and David Tolliver.
\newblock Combinatorial preconditioners and multilevel solvers for problems in
  computer vision and image processing.
\newblock {\em Computer Vision and Image Understanding}, 115(12):1638--1646,
  2011.

\bibitem{DBLP:conf/www/Kunegis13}
J{\'{e}}r{\^{o}}me Kunegis.
\newblock {KONECT:} the koblenz network collection.
\newblock In Leslie Carr, Alberto H.~F. Laender, Bernadette~Farias
  L{\'{o}}scio, Irwin King, Marcus Fontoura, Denny Vrandecic, Lora Aroyo,
  Jos{\'{e}} Palazzo~M. de~Oliveira, Fernanda Lima, and Erik Wilde, editors,
  {\em 22nd International World Wide Web Conference, {WWW} '13, Rio de Janeiro,
  Brazil, May 13-17, 2013, Companion Volume}, pages 1343--1350. International
  World Wide Web Conferences Steering Committee / {ACM}, 2013.
\newblock \href {https://doi.org/10.1145/2487788.2488173}
  {\path{doi:10.1145/2487788.2488173}}.

\bibitem{kyng16}
R.~{Kyng} and S.~{Sachdeva}.
\newblock Approximate gaussian elimination for laplacians - fast, sparse, and
  simple.
\newblock In {\em 2016 IEEE 57th Annual Symposium on Foundations of Computer
  Science (FOCS)}, pages 573--582. {IEEE}, Oct 2016.
\newblock \href {https://doi.org/10.1109/FOCS.2016.68}
  {\path{doi:10.1109/FOCS.2016.68}}.

\bibitem{KyngPenSachdeva16}
Rasmus Kyng, Yin~Tat Lee, Richard Peng, Sushant Sachdeva, and Daniel~A.
  Spielman.
\newblock Sparsified cholesky and multigrid solvers for connection laplacians.
\newblock In {\em Proceedings of the Forty-Eighth Annual ACM Symposium on
  Theory of Computing}, STOC ’16, page 842–850, New York, NY, USA, 2016.
  Association for Computing Machinery.
\newblock \href {https://doi.org/10.1145/2897518.2897640}
  {\path{doi:10.1145/2897518.2897640}}.

\bibitem{DBLP:conf/www/0002PSYZ19}
Huan Li, Richard Peng, Liren Shan, Yuhao Yi, and Zhongzhi Zhang.
\newblock Current flow group closeness centrality for complex networks?
\newblock In Ling Liu, Ryen~W. White, Amin Mantrach, Fabrizio Silvestri,
  Julian~J. McAuley, Ricardo Baeza{-}Yates, and Leila Zia, editors, {\em The
  World Wide Web Conference, {WWW} 2019, San Francisco, CA, USA, May 13-17,
  2019}, pages 961--971. {ACM}, 2019.
\newblock \href {https://doi.org/10.1145/3308558.3313490}
  {\path{doi:10.1145/3308558.3313490}}.

\bibitem{li2018kirchhoff}
Huan Li and Zhongzhi Zhang.
\newblock Kirchhoff index as a measure of edge centrality in weighted networks:
  Nearly linear time algorithms.
\newblock In {\em Proceedings of the Twenty-Ninth Annual ACM-SIAM Symposium on
  Discrete Algorithms}, pages 2377--2396. SIAM, SIAM, 2018.

\bibitem{livne2012lean}
Oren~E Livne and Achi Brandt.
\newblock Lean algebraic multigrid ({LAMG}): Fast graph laplacian linear
  solver.
\newblock {\em SIAM Journal on Scientific Computing}, 34(4):B499--B522, 2012.

\bibitem{Lovasz1996}
L.~Lov\'asz.
\newblock Random walks on graphs: A survey.
\newblock In D.~{Mikl\'os}, V.~T. {S\'os}, and T.~{Sz\H{o}nyi}, editors, {\em
  Combinatorics, Paul Erd\H{o}s is Eighty}, volume~2, pages 353--398. J\'anos
  Bolyai Mathematical Society, Budapest, 1996.

\bibitem{lumsdaine2007challenges}
Andrew Lumsdaine, Douglas Gregor, Bruce Hendrickson, and Jonathan Berry.
\newblock Challenges in parallel graph processing.
\newblock {\em Parallel Processing Letters}, 17(01):5--20, 2007.

\bibitem{10.5555/3086816}
Russell Lyons and Yuval Peres.
\newblock {\em Probability on Trees and Networks}.
\newblock Cambridge University Press, USA, 1st edition, 2017.

\bibitem{magnien2009fast}
Cl{\'e}mence Magnien, Matthieu Latapy, and Michel Habib.
\newblock Fast computation of empirically tight bounds for the diameter of
  massive graphs.
\newblock {\em Journal of Experimental Algorithmics (JEA)}, 13:1--10, 2009.

\bibitem{mavroforakis2015spanning}
Charalampos Mavroforakis, Richard Garcia-Lebron, Ioannis Koutis, and Evimaria
  Terzi.
\newblock Spanning {E}dge {C}entrality: {L}arge-scale {C}omputation and
  {A}pplications.
\newblock In {\em Proceedings of the 24th International Conference on World
  Wide Web, {WWW 2015}}, pages 732--742. ACM, 2015.

\bibitem{mckay1981practical}
Brendan~D McKay et~al.
\newblock {\em Practical graph isomorphism}.
\newblock Department of Computer Science, Vanderbilt University Tennessee, USA,
  1981.

\bibitem{NarayanS18scaling}
Onuttom Narayan and Iraj Saniee.
\newblock Scaling of random walk betweenness in networks.
\newblock In Luca~Maria Aiello, Chantal Cherifi, Hocine Cherifi, Renaud
  Lambiotte, Pietro Li{\'o}, and Luis~M. Rocha, editors, {\em Complex Networks
  and Their Applications VII}, pages 41--51, Cham, 2019. Springer International
  Publishing.

\bibitem{newman2018networks}
Mark Newman.
\newblock {\em Networks (2nd Ed.)}.
\newblock Oxford university press, 2018.

\bibitem{okamoto2008ranking}
Kazuya Okamoto, Wei Chen, and Xiang-Yang Li.
\newblock Ranking of closeness centrality for large-scale social networks.
\newblock In {\em International workshop on frontiers in algorithmics}, pages
  186--195. Springer, Springer, 2008.

\bibitem{pcg2014}
Melissa~E. O'Neill.
\newblock Pcg: A family of simple fast space-efficient statistically good
  algorithms for random number generation.
\newblock Technical Report HMC-CS-2014-0905, Harvey Mudd College, Claremont,
  CA, September 2014.

\bibitem{peng2013efficient}
Richard Peng and Daniel~A. Spielman.
\newblock An efficient parallel solver for sdd linear systems, 2014.
\newblock \href {https://doi.org/10.1145/2591796.2591832}
  {\path{doi:10.1145/2591796.2591832}}.

\bibitem{RanjanZB14incremental}
Gyan Ranjan, Zhi-Li Zhang, and Daniel Boley.
\newblock Incremental computation of pseudo-inverse of laplacian.
\newblock In Zhao Zhang, Lidong Wu, Wen Xu, and Ding-Zhu Du, editors, {\em
  Combinatorial Optimization and Applications}, pages 729--749, Cham, 2014.
  Springer International Publishing.

\bibitem{Schild18}
Aaron Schild.
\newblock An almost-linear time algorithm for uniform random spanning tree
  generation.
\newblock In {\em Proceedings of the 50th Annual ACM SIGACT Symposium on Theory
  of Computing}, STOC 2018, page 214–227, New York, NY, USA, 2018.
  Association for Computing Machinery.
\newblock \href {https://doi.org/10.1145/3188745.3188852}
  {\path{doi:10.1145/3188745.3188852}}.

\bibitem{10.5555/975545}
John Shawe-Taylor and Nello Cristianini.
\newblock {\em Kernel Methods for Pattern Analysis}.
\newblock Cambridge University Press, USA, 2004.

\bibitem{sherman1950}
Jack Sherman and Winifred~J. Morrison.
\newblock Adjustment of an inverse matrix corresponding to a change in one
  element of a given matrix.
\newblock {\em Ann. Math. Statist.}, 21(1):124--127, 03 1950.

\bibitem{Sidje2011}
{Roger B.} Sidje and Yousef Saad.
\newblock Rational approximation to the fermi-dirac function with applications
  in density functional theory.
\newblock {\em Numerical Algorithms}, 56(3):455--479, 3 2011.
\newblock \href {https://doi.org/10.1007/s11075-010-9397-6}
  {\path{doi:10.1007/s11075-010-9397-6}}.

\bibitem{DBLP:journals/siamcomp/SpielmanS11}
Daniel~A. Spielman and Nikhil Srivastava.
\newblock Graph sparsification by effective resistances.
\newblock {\em {SIAM} Journal on Computing}, 40(6):1913--1926, 2011.
\newblock URL: \url{http://dx.doi.org/10.1137/080734029}, \href
  {https://doi.org/10.1137/080734029} {\path{doi:10.1137/080734029}}.

\bibitem{DBLP:journals/netsci/StaudtSM16}
Christian~L. Staudt, Aleksejs Sazonovs, and Henning Meyerhenke.
\newblock Networkit: {A} tool suite for large-scale complex network analysis.
\newblock {\em Network Science}, 4(4):508--530, 2016.
\newblock \href {https://doi.org/10.1017/nws.2016.20}
  {\path{doi:10.1017/nws.2016.20}}.

\bibitem{Tang10aprobing}
Jok~M. Tang and Yousef Saad.
\newblock A probing method for computing the diagonal of a matrix inverse.
\newblock {\em Numerical Linear Algebra with Applications}, 19(3):485--501,
  2012.

\bibitem{van2017pseudoinverse}
Piet Van~Mieghem, Karel Devriendt, and H~Cetinay.
\newblock Pseudoinverse of the laplacian and best spreader node in a network.
\newblock {\em Physical Review E}, 96(3):032311, 2017.

\bibitem{DBLP:conf/hpec/LoozOLM16}
Moritz von Looz, Mustafa~Safa {\"{O}}zdayi, S{\"{o}}ren Laue, and Henning
  Meyerhenke.
\newblock Generating massive complex networks with hyperbolic geometry faster
  in practice.
\newblock In {\em 2016 {IEEE} High Performance Extreme Computing Conference,
  {HPEC} 2016, Waltham, MA, USA, September 13-15, 2016}, pages 1--6. {IEEE},
  2016.
\newblock \href {https://doi.org/10.1109/HPEC.2016.7761644}
  {\path{doi:10.1109/HPEC.2016.7761644}}.

\bibitem{Wilson:1996:GRS:237814.237880}
David~Bruce Wilson.
\newblock Generating random spanning trees more quickly than the cover time.
\newblock In {\em Proceedings of the Twenty-eighth Annual ACM Symposium on
  Theory of Computing}, STOC '96, pages 296--303, New York, NY, USA, 1996. ACM.
\newblock URL: \url{http://doi.acm.org/10.1145/237814.237880}, \href
  {https://doi.org/10.1145/237814.237880} {\path{doi:10.1145/237814.237880}}.

\end{thebibliography}

\newpage

\appendix

\section{Algorithmic Details and Omitted Proofs}

\subsection{Our Approximation Algorithm in More Detail}
\label{app:approx-algo}
\subparagraph{Overall algorithm}
Algorithm~\ref{alg:approx-diag-lpinv} already receives the pivot vertex $u$ as input.
Lines~\ref{line:start-init} to~\ref{line:end-sampling} approximate the effective resistances.
To do so, Lines~\ref{line:start-init} to~\ref{line:end-init} perform initializations:
first, the estimate of the effective resistance is set to $0$ for all vertices. Then the accuracy
$\eta$ of the linear solver is computed so as to ensure an absolute $\epsilon$-approximation for
the whole algorithm. The BFS tree $B_u$ with root $u$ realizes shortest paths
between $u$ and all other vertices. The sample size $\tau$ depends on the parameters $\epsilon$
and $\delta$, among others.

The first \texttt{for}-loop does the actual sampling and aggregation (the latter with Algorithm~\ref{alg:aggregate}).
Afterwards, Lines~\ref{line:start-fill} to~\ref{line:end-fill} fill the $u$-th column and the diagonal of $\Lpinv$ -- to the desired accuracy. Apart from
that, the algorithm's high-level ideas have been provided already in Section~\ref{sub:overview}.

\begin{algorithm}[h!]
  \begin{algorithmic}[1]
    \begin{small}
    \Function{ApproxDiagLpinv}{$G$, $u$, $\epsilon$, $\delta$}
    \State \textbf{Input:} Undirected small-world graph $G = (V, E)$,
      pivot $u \in V$, error bound $\epsilon > 0$, probability $0 < \delta < 1$
    \State \textbf{Output:} $\diag{\widetilde{\Lpinv}}$, \ie an $(\epsilon, \delta)$-approximation of $\diag{\Lpinv}$
    \State $R[v] \gets 0 ~\forall v \in V$ \label{line:start-init}
    \Comment {\small $\Oh(n)$}
    \State Pick constant $\kappa \in (0, 1)$ arbitrarily;
		$\eta \gets \frac{\kappa \epsilon}{3 \sqrt{mn \log n} \diam{G}}$ \label{line:eta}
    \State Compute BFS tree $B_u$ of $G$ with root $u$
    \Comment {\small $\Oh(n + m)$}
    \State $\tau \gets \ecc(u)^2 \cdot \lceil \log(2m/\delta) / (2(1 - \kappa)^2 \epsilon^2) \rceil$
    \Comment {\small $\Oh(1)$} \label{line:end-init}
    \For{$i \gets 1$ to $\tau$}\label{line:ust-sampling-loop}
      \Comment {\small $\tau$ times}
      \State Sample UST $T_i$ of $G$ with root $u$
      \Comment{\small $\Oh(m \log n)$}\label{line:ust-sampling}
      \State $R \gets$ \textsc{Aggregate}($T_i$, $R$, $B_u$)
      \Comment {\small $\Oh(n \log n)$}
    \EndFor \label{line:end-sampling}
    \State Solve $\mat{L} \myvec{x} = \uvec{u} - \frac{1}{n}\cdot \onesvec$ for $\myvec{x} \perp \onesvec$ {\footnotesize (accuracy: $\eta$)} \label{line:start-fill}
    \Comment{\small $\tilde{\Oh}(m \log^{1/2} n \log(1/\eta))$} \label{line:linear-system}
    \State $\ment{\widetilde{\Lpinv}}{u}{u} \gets \vent{x}{u}$
    \Comment {\small $\Oh(1)$}
    \For{$v \in V'$}
    \Comment {\small All iterations: $\Oh(n)$}
      \State $\ment{\widetilde{\mat{L}^\dagger}}{v}{v} \gets R[v] / \tau - \vent{x}{u} + 2 \vent{x}{v}$ \label{line:final-aggregation}
    \EndFor \label{line:end-fill}
    \State \textbf{return} $\operatorname{diag}(\widetilde{\Lpinv})$
    \EndFunction
    \end{small}
  \end{algorithmic}
  \caption{Approximation algorithm for $\diag{\Lpinv}$}
  \label{alg:approx-diag-lpinv}
\end{algorithm}

\begin{remark}
Due to the fact that Laplacian linear solvers provide
a relative error guarantee (and not an absolute
$\pm \epsilon$ guarantee), the (relative)
accuracy $\eta$ for the initial Laplacian linear
system (Lines~\ref{line:eta} and~\ref{line:linear-system})
depends in a non-trivial way on our
guaranteed absolute error $\epsilon$.
For details, see the proofs in Appendix~\ref{sub:app-thm-time}.

We also remark that the value of the constant $\kappa$ does
not affect the asymptotic running time (nor the correctness)
of the algorithm. However, it does affect the empirical
running time by controlling which fraction of the
error budget is invested into solving the initial linear system
vs. UST sampling.
\end{remark}

\if #0
\begin{remark}
  
  \changed{In Line~\ref{line:start-fill}, we require that the solution vector $\myvec{x}$
  of $\mat{L} \myvec{x} = \uvec{u} - \frac{1}{n}\cdot \onesvec$ is perpendicular
  to $\onesvec$, so that $\myvec{x} = \ment{\Lpinv}{:}{u}$ holds. In
  practice we solve an equivalent linear system instead, namely:
  \begin{align}
    \label{eqlinsys}
    (\mat{L} + \frac{1}{n} \mat{J})\myvec{y} & = \uvec{u} - \frac{1}{n}\cdot \onesvec,
  \end{align} where $\mat{J}$ is the $n \times n$-matrix
  with all entries being $1$ and $\myvec{y}$ is the unique solution, $\myvec{y} \in \mathbb{R}^n$. 
  In the following, we prove the above equivalence by showing that
  $\myvec{y} = \ment{\Lpinv}{:}{u}$ for $\myvec{y} \in \mathbb{R}^n$
  and thus $\myvec{y} = \myvec{x}$ for $\myvec{x} \perp \onesvec$.
  }
\changed{\begin{proposition}
  \label{corlinsys}
  The solution vector $\myvec{y} \in \mathbb{R}^n$ of Eq.~(\ref{eqlinsys}) is equal to
  $\ment{\Lpinv}{:}{u}$.
\end{proposition}
}
\changed{\begin{proof} 
  Multiplying both parts of Eq.~(\ref{eqlinsys}) with $\Lpinv + \frac{1}{n} \mat{J}$ we get:
\begin{align}
   (\Lpinv + \frac{1}{n} \mat{J})\cdot (\mat{L} + \frac{1}{n}\mat{J})\myvec{x}
  & = (\Lpinv + \frac{1}{n} \mat{J}) \cdot (\uvec{u} - \frac{1}{n} \cdot \onesvec ) \\
  (\Lpinv \mat{L} + \frac{1}{n} \mat{J} \Lpinv + \frac{1}{n} \mat{J}  \mat{L} + \frac{1}{n^2} \mat{J^2}) \myvec{x}
  & = \Lpinv \uvec{u} - \Lpinv \frac{1}{n}\cdot \onesvec + \frac{1}{n}\mat{J} \uvec{u} -  \frac{1}{n} \mat{J} \cdot \onesvec \label{idem}\\
  ((\mat{I} -\frac{1}{n} \mat{J}) + \mat{O} + \mat{O} + \frac{1}{n} \mat{J}) \myvec{x}
  & = \Lpinv \uvec{u} - \myvec{0} + \frac{1}{n}\cdot \onesvec -  \frac{n}{n^2} \onesvec \\
  \myvec{x}
  & = \Lpinv \uvec{u}, 
\end{align}
where $\mat{O}$ corresponds to the zero matrix and $\myvec{0}$ to the zero vector. Resulting
in zero matrix (or vector) derives from the fact that the row/column sums of both $\mat{L}$
and $\Lpinv$ are equal to zero. In Eq.~(\ref{idem}) we used the fact that $\mat{J^2} = n\mat{J}$. 
\end{proof}
}
\end{remark}
\fi


\subsubsection{Aggregation algorithm}
\begin{algorithm}[h!]
  \begin{algorithmic}[1]
    \begin{small}
    \Function{Aggregate}{$T, R, B_u$}
    \State \textbf{Input:} spanning tree $T$,
      array of effective resistance estimates $R$, shortest-path tree $B_u$
    \State \textbf{Output:} $R$ updated with $T$'s contribution
    \State $\{\alpha, \Omega\} \gets DFS(T)$
        \Comment{$\alpha(v)$, $\Omega(v)$: discovery/finish times of $v$}
        \label{line:tree-dfs}
    \For{$v \in V'$}
        \For{$(a, b) \in P(v) \text{ obtained from } B_u$}
        \If{$parent(b) = a$}%
            \label{line:tree-sign-w}
                \If{$\alpha(b) < \alpha(v)$ \textbf{and} $\Omega(v) < \Omega(b)$}%
\label{line:tree-desc-w}
                    \State $R[v] \gets R[v] + 1$
\label{line:tree-add}
                \EndIf
                \ElsIf{$parent(a) = b$}%
\label{line:tree-sign-wprime}

                \If{$\alpha(a) < \alpha(v)$ \textbf{and} $\Omega(v) < \Omega(a)$}%
\label{line:tree-desc-wprime}
                    \State $R[v] \gets R[v] - 1$
\label{line:tree-sub}
                \EndIf
            \EndIf
        \EndFor
    \EndFor
    \State \textbf{return} $R$
    \EndFunction
    \end{small}
  \end{algorithmic}
  \caption{Aggregation of $T$'s contribution to $R[\cdot]$}
  \label{alg:aggregate}
\end{algorithm}

Algorithm~\ref{alg:aggregate} depicts the pseudocode
of the tree aggregation algorithm
that is discussed in Section~\ref{sub:tree-sampling}.
Here, $\alpha(\cdot)$ and $\Omega(\cdot)$ denote our
DFS discovery and finish timestamps, respectively.
The test whether $(a, b)$ [or $(b, a)$] is in $T$ is carried
out in Line~\ref{line:tree-sign-w} [or Line~\ref{line:tree-sign-wprime}, respectively].
If that is indeed the case,
Line~\ref{line:tree-desc-w} [or Line~\ref{line:tree-desc-wprime}] checks
whether $v$ is below $(a, b)$ [or $(b, a)$, respectively]. If that is the case,
the effective resistance estimate is adapted in Line~\ref{line:tree-add}
[Line~\ref{line:tree-sub}].

\subsubsection{Parallelism}
\label{sub:app:parallelism}
Algorithm~\ref{alg:approx-diag-lpinv} can be parallelized by
sampling and aggregating USTs in parallel.
This yields a work-efficient parallelization in the work-depth model.
The depth of the algorithm is dominated
by (i) computing the BFS tree $B_u$,
(ii) sampling each UST (Line~\ref{line:ust-sampling})
and (iii) solving the Laplacian linear system (Line~\ref{line:linear-system}).
With current algorithms, the latter two procedures have
depth $\Oh(m \log n)$ and $\tilde{\Oh}(m \log^{\frac 12} n \log(1/{\eta}))$, respectively
(simply by executing them sequentially).
We note that parallelizing the loops of Algorithm~\ref{alg:aggregate}
results in a depth of $\Oh(n)$ for Algorithm~\ref{alg:aggregate};
however, this does
not impact the depth of Algorithm~\ref{alg:approx-diag-lpinv}.
We also note that by using a parallel Laplacian solver,
the depth of solving the initial linear system becomes
polylogarithmic~\cite{peng2013efficient}.
Nevertheless,
real-world implementations show a good
parallelization behavior by parallelizing only Algorithm~\ref{alg:approx-diag-lpinv}
(see Sections~\ref{sec:impl} and~\ref{sec:experiments});
consequently, we do not focus on parallelizing the sampling itself.

\subsection{Wilson's UST Algorithm}
\label{app:Wilson}
Given a path $P$, its loop erasure is a simple path created by removing all cycles of $P$ in
chronological order. Wilson's algorithm grows a sequence of sub-trees of $G$, in our case
starting with $u$ as root of $T$. Let $M = \{v_1 ,\ldots, v_{n-1}\}$ be an enumeration of
$V\setminus \{u\}$. Following the order in $M$, a random walk starts from every unvisited $v_i$
until it reaches (some vertex in) $T$ and its loop erasure is added to $T$.
\begin{proposition}[\cite{Wilson:1996:GRS:237814.237880}, comp.~\cite{Hayashi2016EfficientAF}]
\label{thm:algo-Wilson}
For a connected and unweighted undirected graph $G = (V, E)$ and a vertex $u \in V$,
Wilson's algorithm samples a uniform spanning tree of $G$
with root $u$. The expected running time is
the mean hitting time of $G$, $\sum_{v \in V'} \pi_G(v) \kappa_G(v,u)$, where $\pi_G(v)$ is the probability that a random
walk stays at $v$ in its stationary distribution and where $\kappa_G(v,u)$ is the commute
time between $v$ and $u$.
\end{proposition}
\begin{lemma}
\label{lem:ust-time}
Let $G$ be as in Proposition~\ref{thm:algo-Wilson}. Its mean hitting time
can be rewritten as $\sum_{v \in V'} \operatorname{deg}(v) \cdot \effres{u}{v}$, which is $\Oh(\ecc(u) \cdot m)$.
In small-world graphs, this is $\Oh(m \log n)$.
\end{lemma}
\begin{proof} (of Lemma~\ref{lem:ust-time})
First, we replace $\pi_G(v)$ by $\frac{\deg(v)}{\operatorname{vol}(G)}$ in the sum~\cite{Lovasz1996}.
By using the well-known relation $\kappa_G(v,u) = \operatorname{vol}(G) \cdot \effres{v}{u}$~\cite{chandra1996electrical},
the volumes cancel and we obtain $\sum_{v \in V'} \deg(v) \cdot \effres{u}{v}$.
We can bound this from above by $\sum_{v \in V'} \deg(v) \cdot \dist{u, v} \leq \operatorname{vol}(G) \cdot \ecc(u)$,
because effective resistance is never larger than the graph distance~\cite{ericson2013effective}.
In unweighted graphs, $\operatorname{vol}(G) = 2m$ and in undirected graphs $\ecc(u) \leq \diam{G}$ for all $u \in V$,
which proves the claim.
\end{proof}


\subsection{Proof of Lemma~\ref{lem:aggregation-time}}
\label{sub:app-lemma-agg-time}
\begin{proof} (of Lemma~\ref{lem:aggregation-time})
DFS in $T$ takes $\Oh(n)$ time since $T$ is a spanning tree.
Furthermore, Algorithm~\ref{alg:aggregate} loops
over $\Oh(\farnc{G}{u})$ vertex-edge pairs (as $|P(v)| = \dist{u,v}$),
with $\Oh(1)$ query time spent per pair.
Since no path to the root in $B_u$ is longer than $\ecc(u)$,
we obtain $\Oh(n \cdot \ecc(u))$, which is by definition $\Oh(n \cdot \diam{G})$.
\end{proof}

\subsection{Proof of Theorem~\ref{thm:time-complexity}}
\label{sub:app-thm-time}
The proof of our main theorem makes use of Hoeffding's inequality.
In the inequality's presentation, we follow Hayashi \etal~\cite{Hayashi2016EfficientAF}.
\begin{lemma}
\label{lem:hoeffding}
Let $X_1, \dots, X_\tau$ be independent random variables in $[0, 1]$ and
$X = \sum_{i \in [\tau]} X_i$. Then for any $0 < \epsilon < 1$, we have
\begin{equation}
Pr[|X - \mathbb{E}[X]| > \epsilon \tau] \leq 2 \exp(-2 \epsilon^2 \tau).
\end{equation}
\end{lemma}

Before we can prove Theorem~\ref{thm:time-complexity}, we
need auxiliary results on the equivalence of norms.
For this purpose,
let $\Vert \myvec{x} \Vert_{\mat{L}} := \sqrt{\myvec{x}^T \mat{L} \myvec{x}}$
for any $\myvec{x} \in \mathbb{R}^n$. Note that $\Vert \cdot \Vert_{\mat{L}}$
is a norm on the subspace of $\mathbb{R}^n$ with $\myvec{x} \perp \myvec{1}$. We show that:

\begin{lemma}
\label{lem:norms}
Let $G = (V, E)$ be a connected undirected graph with $n$ vertices and $m$ edges. Moreover, let $\mat{L}$ be its Laplacian matrix
and $\lambda_2$ the second smallest eigenvalue of $\mat{L}$.
The volume of $G$, $\operatorname{vol}(G)$,
is the sum of all (possibly weighted) vertex degrees.

For any $\myvec{x} \in \mathbb{R}^n$ with $\myvec{x} \perp \myvec{1}$ we have:
\begin{equation}
\label{eq:norm-equiv}
 \sqrt{\lambda_2} \cdot \Vert \myvec{x} \Vert_\infty \leq \Vert \myvec{x} \Vert_{\mat{L}} \leq \sqrt{2\operatorname{vol}(G)} \cdot \Vert \myvec{x} \Vert_{\infty}.
\end{equation}
\end{lemma}
\begin{proof}
Since $\mat{L}$ is positive semidefinite, it can be seen as a Gram matrix
and written as $\mat{K}^T \mat{K}$ for some real matrix $\mat{K}$. The second smallest eigenvalue
of $\mat{K}$ is then $\sqrt{\lambda_2}$ and we can write:
\begin{equation}
\label{eq:first-part-norms}
  \sqrt{\lambda_2} \cdot \Vert \myvec{x} \Vert_\infty \leq
  \sqrt{\lambda}_2 \cdot \Vert \myvec{x} \Vert_2 =
  \Vert \sqrt{\lambda}_2 \cdot \myvec{x} \Vert_2 \leq
  \Vert \mat{K} \myvec{x} \Vert_2 = \Vert \myvec{x} \Vert_\mat{L}.
\end{equation}
The first, second, and last (in)equality in Eq.~(\ref{eq:first-part-norms})
follow from basic linear algebra facts, respectively.
The third inequality follows from the Courant-Fischer theorem,
since the eigenvector corresponding to the smallest eigenvalue $0$, $\myvec{1}$, is excluded
from the subspace of $\myvec{x}$ (comp. for example Ch.~3.1 of Ref.~\cite{10.5555/975545}.)

Using the quadratic form of the Laplacian matrix, we get:
\begin{align}
  \Vert \myvec{x} \Vert_{\mat{L}} & = \left( \sum_{\{i,j\} \in E} \myvec{w}(u,v) (\vent{x}{i} - \vent{x}{j})^2 \right)^{1/2} \\
  & \leq \left( \frac{1}{2} \operatorname{vol}(G) \cdot (2 \Vert \myvec{x} \Vert_{\infty})^2   \right)^{1/2} = \sqrt{2\operatorname{vol}(G)} \cdot \Vert \myvec{x} \Vert_\infty
\end{align}
\end{proof}

We are now in the position to prove our main result:

\begin{proof} (of Theorem~\ref{thm:time-complexity})
Solving the initial linear system with the solver
by Cohen \etal~\cite{CohenKyng14}
takes $\tilde{\Oh}(m \log^{1/2} n \cdot \log (1/\eta))$ time to achieve a relative error bound
of $\Vert \tilde{\myvec{x}} - \myvec{x} \Vert_{\mat{L}} \leq \eta \Vert \myvec{x} \Vert_\mat{L}$. Here, $\myvec{x}$
is the true solution, $\tilde{\myvec{x}}$ the estimate, and
$\Vert \myvec{x} \Vert_{\mat{L}} = \sqrt{\myvec{x}^T \mat{L} \myvec{x}}$.
To make this error bound compatible with the absolute error we pursue,
we first note that $\sqrt{\lambda_2} \cdot \Vert \tilde{\myvec{x}} \Vert_\infty \leq \Vert \tilde{\myvec{x}} \Vert_{\mat{L}} \leq \sqrt{2 \operatorname{vol}(G)} \cdotp \Vert \tilde{\myvec{x}} \Vert_\infty$
(Lemma~\ref{lem:norms}), where $\lambda_2$ is the second smallest eigenvalue of $\mat{L}$.
We may use Lemma~\ref{lem:norms}, as
$\tilde{\myvec{x}}$ and $\myvec{x}$ are both perpendicular to $\myvec{1}$ (since the
image of $\mat{L}$ is perpendicular to its kernel, which is $\myvec{1}$).
It is known that $\lambda_2 \geq 4 / (n \cdot \diam{G})$~\cite{mckay1981practical}.
Hence, if we set $\eta := \frac{\kappa \epsilon}{3 \cdot \sqrt{mn \log n} \diam{G}}$,
we get for small-world graphs:
\begin{align*}
  \Vert \tilde{\myvec{x}} - \myvec{x} \Vert_\infty
  & \leq \frac{1}{\sqrt{\lambda_2}} \cdot \Vert \tilde{\myvec{x}} - \myvec{x} \Vert_\mat{L}
  \leq \frac{\eta}{\sqrt{\lambda_2}} \cdot \Vert \myvec{x} \Vert_\mat{L} \\
  & \leq \frac{\eta}{\sqrt{4/(n \cdot \diam{G}})} \cdot \Vert \myvec{x} \Vert_\mat{L}
  \leq \frac{\eta \sqrt{n \log n}}{2} \cdot 2\sqrt{m}\cdot \Vert \myvec{x} \Vert_\infty \\
  & = \frac{\kappa \epsilon}{3 \sqrt{mn \log n} \diam{G}} \cdot \sqrt{mn \log n} \cdot \Vert \myvec{x} \Vert_\infty \\
  & \leq   \frac{\kappa \epsilon}{3 \diam{G}} \Vert \myvec{x} \Vert_\infty
  \leq \frac{\kappa \epsilon}{3 \diam{G}} \diam{G}
  \leq \frac{\kappa}{3} \epsilon.
\end{align*}
The second last inequality follows from the fact that $\myvec{x}$ expresses potentials scaled by $1/n$,
arising from $n-1$ (scaled) effective resistance problems fused
together. The maximum norm of $\myvec{x}$ can thus be bounded
by $(n-1) \frac{1}{n} \effres{u}{v} \leq \diam{G}$,
because $\effres{\cdot}{\cdot}$ is bounded by the graph distance.

Taking Eq.~(\ref{eq:diag-comp}) into account,
this means that the maximum error of a diagonal value in $\widetilde{\Lpinv}$ as a consequence from the linear system can be bounded
by $\kappa \epsilon$.
The resulting running time for the solver is then $\tilde{\Oh}(m \log^{1/2} n
\log (\frac{3 \sqrt{mn \log n}\diam{G}}{\kappa \epsilon})) = \tilde{\Oh}(m \log^{1/2} n \log(n/\epsilon))$.

The main bottleneck is the loop that samples $\tau$ USTs and aggregates their contribution in each iteration.
According to Lemma~\ref{lem:ust-time}, sampling takes $\Oh(m \log n)$ time per tree
in small-world graphs. Aggregating a tree's contribution is less expensive
(Lemma~\ref{lem:aggregation-time}).

Let us determine next the sample size $\tau$ that allows the desired guarantee.
To this end, let $\epsilon' := (1 - \kappa)\epsilon$
denote the possible absolute error for the effective resistance estimates.
Plugging $\tau := \ecc(u)^2 \cdot \lceil \log(2m/\delta) / (2(\epsilon')^2) \rceil$
into Hoeff\-ding's inequality (Lemma~\ref{lem:hoeffding}), yields for each single edge $e \in E$
and its estimated electrical flow $\myvec{\tilde{f}}(e)$: $Pr[\myvec{\tilde{f}}(e) = \myvec{f}(e) \pm \epsilon' / \ecc(u)] \geq 1 - \delta / m$.
Using the union bound, we get that $Pr[\myvec{\tilde{f}}(e) = \myvec{f}(e) \pm \epsilon' / \ecc(u)]$
for all $e \in E$ at the same time holds with probability $\geq 1 - \delta$. Since for all $v$ the path length
$|P(v)|$ is bounded by $\ecc(u)$, another application of the union bound yields that
$Pr[R[v] = \effres{u}{v} \pm \epsilon'] \geq 1 - \delta$.
\end{proof}


\subsection{Proof of Lemma~\ref{lem:nrwb-derivation}}
\label{app:nrwb}
\begin{proof}
Recall that the normalized random-walk betweenness is expressed as follows (Eq.~(\ref{normalizerandomwalk})):
\begin{align*}
  c_b(v) & = \frac{1}{n}+\frac{1}{n-1} \sum_{t \neq v} \frac{\ment{M^{-1}}{t}{t} - \ment{M^{-1}}{t}{v}}{\ment{M^{-1}}{t}{t} + \ment{M^{-1}}{v}{v} -2\ment{M^{-1}}{t}{v}}
\end{align*}
where $\mat{M} := \mat{L} + \mat{P}$ with $\mat{L}$ the Laplacian matrix
and $\mat{P}$ the projection
operator onto the zero eigenvector of the Laplacian such that $\ment{P}{i}{j} = 1/n$.
We also have $\Lpinv := (\mat{L} + \mat{P})^{-1} - \mat{P}$ and thus we can
replace $\mat{M}^{-1}$ with $\Lpinv + \mat{P}$.
Then, for the numerator of  Eq.~(\ref{normalizerandomwalk}) we have:
\begin{align}
  \label{numerator}
  & \sum_{t \neq v} (\ment{M^{-1}}{t}{t} - \ment{M^{-1}}{t}{v}) = \nonumber \\
  & \sum_{t \neq v} (\ment{\Lpinv}{t}{t} - \ment{P}{t}{t} -\ment{\Lpinv}{t}{v} + \ment{P}{t}{v}) = \nonumber \\
  & \sum_{t \neq v} (\ment{\Lpinv}{t}{t} -\ment{\Lpinv}{t}{v}) = \trace{\Lpinv} - \ment{\Lpinv}{v}{v} - \sum_{t \neq v} \ment{\Lpinv}{t}{v} = \nonumber\\
  & \trace{\Lpinv}.
\end{align}
The second equality holds because $ \sum\limits_{l \neq v} (\ment{P}{t}{v} - \ment{P}{t}{t}) = 0$
for all $t,v \in V$.
The final equality holds since $\sum\limits_{t \neq v}^n\ment{\Lpinv}{t}{v} = -\ment{\Lpinv}{v}{v}$
for all $v \in V$.

Then, for the denominator we have:
  \begin{align}
    \label{denom}
    & (n-1) \sum\limits_{t \neq v} (\ment{M^{-1}}{t}{t} + \ment{M^{-1}}{v}{v} -2\ment{M^{-1}}{t}{v}) = \nonumber \\
    & (n-1) \sum\limits_{t \neq v} (\ment{\Lpinv}{t}{t} + \ment{P}{t}{t} + \ment{\Lpinv}{v}{v}
    +  \ment{P}{v}{v} -2\ment{\Lpinv}{t}{v} - 2 \ment{P}{t}{v})  = \nonumber \\
    & (n-1) \sum\limits_{t \neq v} (\ment{\Lpinv}{t}{t} + \ment{\Lpinv}{v}{v} -2\ment{\Lpinv}{t}{v}) = \nonumber \\
    & (n-1)\farnel{G}{v}.
  \end{align}
  The second equality holds since $\sum\limits_{t \neq v}\ment{P}{t}{t} +\ment{P}{v}{v} -2 \ment{P}{t}{v} = 0$ for all $t, v \in V$ and the last equality due to Eq.~(\ref{eq:eff-res}) and the definition
  of electrical farness.
  The claim follows from combining Eqs.~(\ref{numerator}) and~(\ref{denom}).
\end{proof}

\section{Kirchhoff Index and Related Centralities}
\label{app:Kirchhoff}
\subsection{Description}
\label{sub:app:Kirchhoff-descr}
The sum of the effective resistance distances over all pairs of vertices is an
important measure for network robustness known as the Kirchhoff index $\mathcal{K}(G)$
or (effective) graph resistance~\cite{Klein93,Ellens2011}.
The Kirchhoff index is often computed via the closed-form expression $\mathcal{K}(G) = n \trace{\Lpinv}$~\cite{Klein93},
where the trace is the sum of the diagonal elements.
Li and Zhang~\cite{li2018kirchhoff} recently
adapted the Kirchhoff index to obtain two edge centrality measures for $e \in E$:
(i)
$\mathcal{C}_{\theta}(e) := n\trace{\Lpinv\setminus_{\theta}e}$, where
$\mat{L}\setminus_{\theta}e$ corresponds to a graph in which edge $e$ has been down-weighted
according to a parameter $\theta$ and (ii) $\mathcal{C}^{\Delta}_{\theta}(e) := \mathcal{C}_{\theta}(e) - \mathcal{K}(G)$,
which quantifies the difference of the Kirchhoff indices
between the new and the original graph.

\subsection{Related Work}
\label{sub:app:Kirchhoff-rel-work}
To calculate the Kirchhoff edge centralities, Ref.~\cite{li2018kirchhoff}
uses techniques such as partial Cholesky factorization~\cite{kyng16},
fast Laplacian solvers and the Hutchinson estimator.
For $\mathcal{C}^{\Delta}_{\theta}(e)$, which is the more interesting measure in our context,
they propose an $\epsilon$-approximation
algorithm that approximates $\mathcal{C}^{\Delta}_{\theta}(e)$ for all edges in
$\Oh(m\epsilon^{-2}\theta^{-2}\log^{2.5}{n}\log(1/\epsilon))$ time (up to polylogarithmic factors).
The algorithm uses the Sherman-Morrison formula~\cite{sherman1950}, which
gives a fractional expression of $(\Lpinv\setminus_{\theta}e - \Lpinv)$.
The numerator is approximated by the Johnson-Lindenstrauss lemma,
and the denominator by effective resistance estimates for all edges.

\subsection{Kirchhoff Index and Edge Centralities}
\label{sub:app:Kirchhoff-results}
It is easy to see that Algorithm~\ref{alg:approx-diag-lpinv} can approximate Kirchhoff Index,
exploiting the expression $\mathcal{K}(G) = n \trace{\Lpinv}$~\cite{Klein93}.
As a direct consequence, we have:

\begin{proposition}
  \label{KirchhoffIndex}
  Let $G$ be a small-world graph as in Theorem~\ref{thm:time-complexity}.
  Then, Algorithm~\ref{alg:approx-diag-lpinv} approximates with high probability
  $\mathcal{K}(G)$ with absolute error $\pm \epsilon$ in $\Oh(m \log^4 n \cdot \epsilon^{-2})$ time.
\end{proposition}

We also observe that we can use a component of Algorithm~\ref{alg:approx-diag-lpinv} to approximate
$\mathcal{C}^{\Delta}_{\theta}(e)$.
Recall that $\mathcal{C}^{\Delta}_{\theta}(e) = \mathcal{C}_{\theta}(e) - \mathcal{K}(G) = n(\trace{\Lpinv\setminus_{\theta}e} - \trace{\Lpinv})$.
Using the Sherman-Morrison formula, as done in Ref.~\cite{li2018kirchhoff}, we have:
  \begin{align}
    \label{eq:ck}
    \mathcal{C}^{\Delta}_{\theta}(e)= n(1-\theta)\frac{\myvec{w}(e)\trace{\Lpinv \myvec{b}_e\myvec{b}^\top_e\Lpinv}}{1-(1-\theta)\myvec{w}(e)\myvec{b}^\top_e\Lpinv \myvec{b}_e},
  \end{align}
where $\myvec{b}_e$ for $e = (u,v)$ is the vector $\uvec{u} - \uvec{v}$.

Ref.~\cite{li2018kirchhoff} approximates $\mathcal{C}^{\Delta}_{\theta}(e)$ with an
algorithm that runs in $\Oh(m\theta^{-2}\log^{2.5}{n}\log(1/\epsilon)\poly{\log\log n} \cdot \epsilon^{-2})$ time. The algorithm is dominated by
the denominator of Eq.~(\ref{eq:ck}), which runs in
$\Oh(m \theta^{-2}\log^{2.5}{n} \poly{\log\log n} \cdot \epsilon^{-2})$.
For the numerator of Eq.~(\ref{eq:ck}), they use the following Lemma:
\begin{lemma}
\label{KC-numerator}
(paraphrasing from Ref.~\cite{li2018kirchhoff})
Let $\mat{L}$ be a Laplacian matrix and $\epsilon$ a scalar such that $0 < \epsilon \leq 1/2 $.
There is an algorithm that achieves an $\epsilon$-approximation of the numerator of Eq.~(\ref{eq:ck})
with high probability in $\Oh(m \log^{1.5}{n} \log(1/\epsilon)  \cdot \epsilon^{-2})$ time.
\end{lemma}

The algorithm in Lemma~\ref{KC-numerator} uses the Monte-Carlo estimator
with $\Oh(\epsilon^{-2} \log{n})$ random vectors $\myvec{z}_i$ to calculate the trace
of the implicit matrix $\myvec{y}^{\top}_i \myvec{b}_e \myvec{b}^{\top}_e\myvec{y}_i$,
where $\myvec{y}_i$ is the approximate solution of $\myvec{y}_i := \Lpinv\myvec{z}_i$ -- derived from solving the corresponding linear system involving $\mat{L}$.
For each system, the Laplacian solver runs in $\Oh(m \log^{1/2}{n} \log(1/\epsilon))$ time.

We notice that a UST-based sampling approach works again for the denominator:
The denominator is just $1-(1-\theta)\myvec{w}(e)\myvec{r}(e)$, where $e \in E$
($\myvec{r}(e) = \myvec{b}^\top_e\Lpinv \myvec{b}_e$).
Approximating $\myvec{r}(e)$ for every $e \in E$ then requires sampling
USTs and counting for each edge $e$ the number of USTs it appears in.
Moreover, we only need to sample  $q = \lceil 2\epsilon^{-2}\log(2m/\delta) \rceil$ to
get an $\epsilon$-approximation
of the effective resistances for all edges
(using Theorem~8 in Ref.~\cite{Hayashi2016EfficientAF}).
Since $\myvec{r}(e)$ are approximate, we need to bound their approximation
when subtracted from $1$. Following Ref.~\cite{li2018kirchhoff},
we use the fact that $0 < \theta <1$ and that for each edge
$\myvec{w}(e)\myvec{r}(e)$ is between $0$ and $1$, bounding the denominator.
The above algorithm can be used to approximate the denominator
of Eq.~(\ref{eq:ck}) with absolute error $\pm \epsilon$
in $\Oh(m \log^2 n \cdot \epsilon^{-2})$ time. Combining the above algorithm
and Lemma~\ref{KC-numerator}, it holds that:

\begin{proposition}
  \label{KirchhoffCentral}
Let $G = (V, E)$ be a small-world graph as in Theorem~\ref{thm:time-complexity}.
Then, there is an algorithm (using Lemma~\ref{KC-numerator} and our Wilson-based sampling algorithm)
that approximates with high probability $\mathcal{C}^{\Delta}_{\theta}(e)$
for all $e \in E$ with absolute error $\pm \epsilon$
in $\Oh(m \log^2 n \log(1/\epsilon) \cdot \epsilon^{-2})$ time.
\end{proposition}

\section{Detailed Engineering Aspects}
\subsection{UST Generation, Pivot Selection, and the Linear System}
\label{sub:general-engineering}
Wilson's algorithm~\cite{Wilson:1996:GRS:237814.237880} using loop-erased
random walks is the best choice in practice for UST generation and also
the fastest asymptotically for unweighted small-world graphs.
A fast random number generator is required for this algorithm;
our code uses PCG32~\cite{pcg2014} for this purpose.
For our implementation we use a variant of Wilson's algorithm
to sample each tree, proposed by Hayashi \etal~\cite{Hayashi2016EfficientAF}:
first, one computes the biconnected components of $G$, then applies
Wilson to each biconnected component, and finally combines the component trees to a UST of $G$.
In each component, we use a vertex with maximal degree
	as root for Wilson's algorithm.
Using this approach, Hayashi \etal~\cite{Hayashi2016EfficientAF} experienced an average
performance improvement of around 40\% on sparse graphs compared
to running Wilson directly.

As a consequence of Theorem~\ref{thm:time-complexity}, the pivot vertex $u$ should be chosen
to have low eccentricity. As finding the
vertex with lowest eccentricity with a naive APSP approach would be too
expensive, we compute a lower bound on the eccentricity for all vertices of
the graph and choose $u$ as the vertex with the lowest bound. These bounds are
computed using a strategy analogous to the double sweep lower bound
by Magnien \etal~\cite{magnien2009fast}: we run a BFS from a random vertex $v$,
then another BFS from the farthest vertex from $v$, and so on.
At each BFS we update the bounds of all the visited vertices; an empirical
evaluation has shown that 10 iterations yield a reasonably accurate approximation of
the vertex with lowest eccentricity.

As a result from preliminary experiments, we use a general-purpose Conjugate Gradient (CG)
solver for the single (sparse) Laplacian linear systems, together with a diagonal
preconditioner. We choose the implementation of the C++ library Eigen~\cite{eigenweb}
for this purpose and found that the accuracy parameter
$\kappa = 0.3$ yields a good trade-off
between the CG and UST sampling steps.

\subsection{Parallel Implementation}
\label{sec:app:parallelism}


\subparagraph{Shared memory}

Our implementation uses OpenMP for shared-memory parallelism.
We aggregate $R[\cdot]$ in thread-local vectors
and perform a final parallel reduction
over all $R[\cdot]$. We found that on the graphs that we can
handle in shared memory, no sophisticated load balancing strategies
are required to achieve reasonable scalability.


\subparagraph{Distributed memory}

We provide an implementation of our algorithm
for replicated graphs in distributed memory
that exploits hybrid parallelism based on MPI + OpenMP.
On each compute node, we take samples and aggregate
$R[\cdot]$ as in shared memory.
Compared to the shared-memory implementation, however,
our distributed-memory implementation exhibits two main peculiarities:
(i) we still solve the initial Laplacian system on a single compute node
only; we interleave, however, this step with UST sampling on other compute nodes,
and (ii) we employ explicit load balancing.
The choice to solve the initial system on a single compute node only
is done to avoid additional communication among nodes.
In fact, we only expect distributed CG solvers to outperform
this strategy for inputs that are considerably larger than
the largest graphs that we consider.
Furthermore, since we interleave this step of the algorithm with UST sampling
on other compute nodes, our strategy only results in a bottleneck
on input graphs where solving a single Laplacian system is slower
than taking \emph{all} UST samples -- but these inputs are
already \enquote{easy}.

For load balancing, the naive approach would consist of
statically taking $\ceil{\tau/p}$ UST samples on each of the $p$ compute nodes.
However, in contrast to the shared-memory case, this does not yield
satisfactory scalability.
In particular, for large graphs, the running time of the UST
sampling step has a high variance.
To alleviate this issue, we use a simple \emph{dynamic} load balancing
strategy:
periodically, we perform an asynchronous reduction
(\texttt{MPI\_Iallreduce}) to calculate the total number of
UST samples taken so far (over all compute nodes).
Afterwards, each compute node calculates the number of samples
that it takes before the next asynchronous reduction
(unless more than $\tau$ samples
were taken already, in which case the algorithm stops).
We compute this number as $\lceil \tau / (b \cdot p^\xi) \rceil$
for fixed constants $b$ and $\xi$.
We also overlap the asynchronous reduction with
additional sampling to avoid idle times.
Finally, we perform a synchronous reduction (\texttt{MPI\_Reduce}) to aggregate
$R[\cdot]$ on a single compute node before outputting the resulting diagonal values.
By parameter tuning~\cite{angriman2019guidelines},
we found that choosing $b = 25$ and $\xi = 0.75$
yields the best parallel scalability.

\newpage

\ifthenelse{\boolean{confversion}}{}{
\section{Experimental Settings}

\subsection{Excluded Competitors}
\label{sub:app:excluded}
PSelInv~\cite{JACQUELIN201884} is a distributed-memory tool
for computing selected elements of $\mat{A}^{-1}$ --
exactly those that correspond to the non-zero entries of the original matrix $\mat{A}$.
However, when a smaller set of elements is required
(such as $\diag{\Lpinv}$), PSelInv is not competitive
on our input graphs: preliminary experiments of ours
have shown that even on $4\times 24$ cores PSelInv is
one order of magnitude slower than a sequential run of our algorithm.

Another conceivable way to compute $\diag{\Lpinv}$ is to extract
the diagonal from a low-rank approximation of $\Lpinv$~\cite{bozzo2012approximations}
using a few eigenpairs. However, our experiments have shown that
this method is not competitive -- neither in terms of quality nor in running time.

Hence, we do not include Refs.~\cite{JACQUELIN201884,bozzo2012approximations} in the presentation of our experiments.

\subsection{Instance Statistics}
\label{sub:app:instances}
Tables~\ref{tab:networks-medium-gt}, \ref{tab:networks-medium},
\ref{tab:networks-large}, and \ref{tab:networks-cluster}
depict detailed statistics about the real-world instances
used in our experiments.

\begin{table}[h]
\setlength{\tabcolsep}{\instTabColSep}
\small
\centering
\caption{Medium-size instances with ground truth}
\label{tab:networks-medium-gt}
\begin{tabular}{lrrrrrr}
Network & Type & ID & $|V|$ & $|E|$ & diam & ecc($u$)\\
\midrule
slashdot-zoo & social & \texttt{sz} & \numprint{79116} & \numprint{467731} & \numprint{12} & \numprint{6}\\
petster-cat-household & social & \texttt{pc} & \numprint{68315} & \numprint{494562} & \numprint{10} & \numprint{6}\\
wikipedia\_link\_ckb & web & \texttt{wc} & \numprint{60257} & \numprint{801794} & \numprint{13} & \numprint{7}\\
wikipedia\_link\_fy & web & \texttt{wf} & \numprint{65512} & \numprint{921533} & \numprint{10} & \numprint{5}\\
loc-gowalla\_edges & social & \texttt{lg} & \numprint{196591} & \numprint{950327} & \numprint{16} & \numprint{8}\\
petster-dog-household & social & \texttt{pd} & \numprint{255968} & \numprint{2148090} & \numprint{11} & \numprint{6}\\
livemocha & social & \texttt{lm} & \numprint{104103} & \numprint{2193083} & \numprint{6} & \numprint{4}\\
petster-catdog-household & social & \texttt{pa} & \numprint{324249} & \numprint{2642635} & \numprint{12} & \numprint{7}\\
\midrule
\end{tabular}

\end{table}

\begin{table}[h]
\setlength{\tabcolsep}{\instTabColSep}
\small
\centering
\caption{Medium-sized instances without ground truth}
\label{tab:networks-medium}
\begin{tabular}{lrrrrrr}
Network & Type & ID & $|V|$ & $|E|$ & diam & ecc($u$)\\
\midrule
eat & words & \texttt{ea} & \numprint{23132} & \numprint{297094} & \numprint{6} & \numprint{4}\\
web-NotreDame & web & \texttt{wn} & \numprint{325729} & \numprint{1090108} & \numprint{46} & \numprint{23}\\
citeseer & citation & \texttt{cs} & \numprint{365154} & \numprint{1721981} & \numprint{34} & \numprint{18}\\
wikipedia\_link\_ml & web & \texttt{wm} & \numprint{131288} & \numprint{1743937} & \numprint{12} & \numprint{7}\\
wikipedia\_link\_bn & web & \texttt{wb} & \numprint{225970} & \numprint{2183246} & \numprint{11} & \numprint{6}\\
flickrEdges & images & \texttt{fe} & \numprint{105722} & \numprint{2316668} & \numprint{9} & \numprint{6}\\
petster-dog-friend & social & \texttt{pr} & \numprint{426485} & \numprint{8543321} & \numprint{11} & \numprint{7}\\
\midrule
\end{tabular}

\end{table}

\begin{table}[h]
\setlength{\tabcolsep}{\instTabColSep}
\small
\centering
\caption{Large instances}
\label{tab:networks-large}
\begin{tabular}{lrrrrr}
Network & Type & $|V|$ & $|E|$ & diam & ecc($u$)\\
\midrule
hyves & social & \numprint{1402673} & \numprint{2777419} & \numprint{10} & \numprint{7}\\
com-youtube & social & \numprint{1134890} & \numprint{2987624} & \numprint{24} & \numprint{12}\\
flixster & social & \numprint{2523386} & \numprint{7918801} & \numprint{8} & \numprint{4}\\
petster-catdog-friend & social & \numprint{575277} & \numprint{13990793} & \numprint{13} & \numprint{7}\\
flickr-links & social & \numprint{1624991} & \numprint{15473043} & \numprint{24} & \numprint{12}\\
\midrule
\end{tabular}

\end{table}

\begin{table}[h]
\setlength{\tabcolsep}{\instTabColSep}
\small
\centering
\caption{Large instances used only on $16\times 24$ cores.}
\label{tab:networks-cluster}
\begin{tabular}{lrrrrr}
Network & Type & $|V|$ & $|E|$ & diam & ecc($u$)\\
\midrule
petster-carnivore & social & \numprint{601213} & \numprint{15661775} & \numprint{15} & \numprint{8}\\
soc-pokec-relationships & social & \numprint{1632803} & \numprint{22301964} & \numprint{14} & \numprint{8}\\
soc-LiveJournal1 & social & \numprint{4843953} & \numprint{42845684} & \numprint{20} & \numprint{10}\\
livejournal-links & social & \numprint{5189808} & \numprint{48687945} & \numprint{23} & \numprint{12}\\
orkut-links & social & \numprint{3072441} & \numprint{117184899} & \numprint{10} & \numprint{6}\\
wikipedia\_link\_en & web & \numprint{13591759} & \numprint{334590793} & \numprint{12} & \numprint{7}\\
\midrule
\end{tabular}

\end{table}

\newpage

\subsection{Relative Error Quality Measures}
\label{sub:app:rel-err-qual}

Because our algorithm computes an absolute $\pm \epsilon$-approximation of $\diag{\Lpinv}$
with high probability, it is expected to yield better results in terms of
maximum absolute error and ranking than numerical approaches with a relative
error guarantee. Indeed, as we show in Appendix~\ref{app:additional_experiments},
the quality assessment changes if we consider
quality measures based on a relative error such as:
\[
    \lonerel := \frac{\onen{\diag{\Lpinv} - \diag{\widetilde{\Lpinv}}}}{\onen{\diag{\Lpinv}}},
\]
\[
    \ltworel := \frac{\twon{\diag{\Lpinv} - \diag{\widetilde{\Lpinv}}}}{\twon{\diag{\Lpinv}}},
\]
and
\[
    \erel := \mathrm{gmean}_i \frac{|\Lpinv_{ii} - \widetilde{\Lpinv_{ii}}|}{\Lpinv_{ii}}.
\]

\section{Additional Experimental Results}
\label{app:additional_experiments}

\begin{figure}[htp]
\centering
\begin{subfigure}[t]{\textwidth}
\centering
\includegraphics{plots/legend_quality}
\end{subfigure}
\begin{subfigure}[t]{.33\textwidth}
\centering
\includegraphics{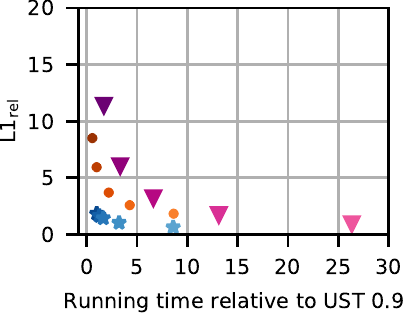}
\end{subfigure}\hfill
\begin{subfigure}[t]{.33\textwidth}
\centering
\includegraphics{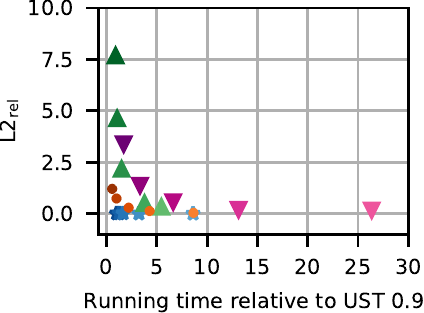}
\end{subfigure}\hfill
\begin{subfigure}[t]{.33\textwidth}
\centering
\includegraphics{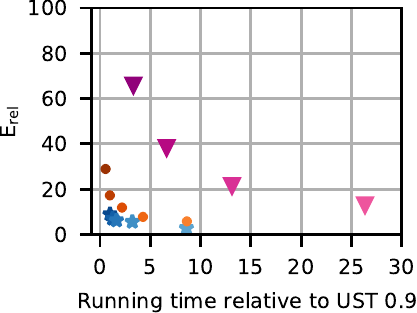}
\end{subfigure}
\caption{$\lonerel$, $\ltworel$ and $\erel$ \wrt the running time of our algorithm with
$\epsilon = 0.9$. All data points are aggregated using the geometric mean over
the instances of Table~\ref{tab:networks-medium-gt}.}
\label{fig:rel-errs}
\end{figure}

Figure~\ref{fig:rel-errs} shows that, when assessing the error in terms of $\lonerel$,
$\ltworel$, or $\erel$, for the same running time
\ust yields results that are still better in terms of quality than the
competitors', but not by such a wide margin.
This can be explained by the fact that
the numerical solvers used by our competitors often employ measures analogous to $\lonerel$ and
$\ltworel$ in their stopping conditions.

\newpage

\subsection{Parallel Scalability}
\label{sub:app:par-scal}

\parallelScalability

\timeBreakdown
Figures~\ref{fig:par-scal} and~\ref{fig:breakdown} report additional results
regarding parallel scalability of \ust.

\subsection{Scalability on R-MAT Graphs}
\label{sub:rmat}

In Figure~\ref{fig:rmat} we report additional results about the scalability of \ust \wrt the
graph size using the R-MAT~\cite{chakrabarti2004r} model. For this
experiment we use the Graph500 parameter setting (\ie edge factor 16, $a =
0.57$, $b = 0.19$, $c = 0.19$, and $d = 0.05$).
The algorithm
requires only \maxTimeRmat minutes on inputs with up to \maxEdgesRmat million edges.
In particular, since these graphs have a nearly-constant diameter,
our algorithm is faster than on random hyperbolic graphs.
Qualitatively, it exhibits a similar scalability.

\syntheticInstances

\pagebreak

\subsection{Memory Consumption}
\label{sec:memory}

Finally, we measure the peak memory consumption of all the algorithms while
running sequentially on the instances of Tables~\ref{tab:networks-medium-gt}
and~\ref{tab:networks-medium}.
More precisely, we subtract the peak resident set size before launching the
algorithm from the peak resident set size after the algorithm finished.
Figure~\ref{fig:memory} shows that \ust requires
less memory than the competitors on all the considered instances.
This can be explained by the fact that, unlike its competitors, our algorithm
does not rely on Laplacian solvers with considerable
memory overhead.
For the
largest network in particular, the peak memory is \ustMemPR MB for \ust, and at
least \lamgMemPR GB for the competitors.

\begin{figure}[tb]
\centering
\includegraphics{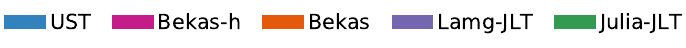}
\includegraphics{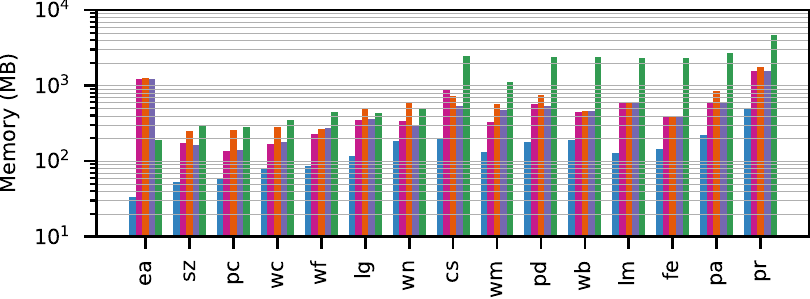}
\caption{Difference between the peak resident set size before and after a sequential
run of each algorithm on the instances of Tables~\ref{tab:networks-medium-gt}
and~\ref{tab:networks-medium}.}
\label{fig:memory}
\end{figure}

\newpage

\subsection{Baseline}
\label{app:baseline}

\begin{table}[h]
    \centering
    \caption{
    Precision of the diagonal entries computed by the LAMG solver (tolerance: $\lamgTol$) compared with
    the ones computed by the Matlab \texttt{pinv} function.}
    \label{tab:lamg_pinv_table}

  \begin{adjustbox}{width=\textwidth}
    \begin{tabular}{lrrrrrrrrr}
Network & Type & $|V|$ & |E| & diam. & $\max_i \Lpinv_{ii} - \widetilde{\Lpinv_{ii}}$ & $\erel$ & $\lonerel$ & $\ltworel$ & Ranking\\
\midrule
moreno-lesmis & characters & \numprint{77} & \numprint{254} & \numprint{5} & \numprint{0.0000} & \numprint{0.00}\% &\numprint{0.00}\% & \numprint{0.00}\% & \numprint{0.48}\%\\
petster-hamster-household & social & \numprint{874} & \numprint{4003} & \numprint{8} & \numprint{0.0006} & \numprint{0.23}\% &\numprint{0.13}\% & \numprint{0.07}\% & \numprint{0.02}\%\\
subelj-euroroad & infrastructure & \numprint{1039} & \numprint{1305} & \numprint{62} & \numprint{0.0031} & \numprint{0.12}\% &\numprint{0.09}\% & \numprint{0.05}\% & \numprint{0.00}\%\\
arenas-email & communication & \numprint{1133} & \numprint{5451} & \numprint{8} & \numprint{0.0002} & \numprint{0.13}\% &\numprint{0.07}\% & \numprint{0.03}\% & \numprint{0.00}\%\\
dimacs10-polblogs & web & \numprint{1222} & \numprint{16714} & \numprint{8} & \numprint{0.0002} & \numprint{0.18}\% &\numprint{0.07}\% & \numprint{0.02}\% & \numprint{0.01}\%\\
maayan-faa & infrastructure & \numprint{1226} & \numprint{2408} & \numprint{17} & \numprint{0.0005} & \numprint{0.08}\% &\numprint{0.06}\% & \numprint{0.03}\% & \numprint{0.00}\%\\
petster-hamster-friend & social & \numprint{1788} & \numprint{12476} & \numprint{14} & \numprint{0.0003} & \numprint{0.15}\% &\numprint{0.07}\% & \numprint{0.02}\% & \numprint{0.01}\%\\
petster-hamster & social & \numprint{2000} & \numprint{16098} & \numprint{10} & \numprint{0.0001} & \numprint{0.09}\% &\numprint{0.04}\% & \numprint{0.02}\% & \numprint{0.01}\%\\
wikipedia-link-lo & web & \numprint{3733} & \numprint{82977} & \numprint{9} & \numprint{0.0001} & \numprint{0.05}\% &\numprint{0.02}\% & \numprint{0.01}\% & \numprint{0.03}\%\\
advogato & social & \numprint{5042} & \numprint{39227} & \numprint{9} & \numprint{0.0001} & \numprint{0.03}\% &\numprint{0.02}\% & \numprint{0.01}\% & \numprint{0.01}\%\\
p2p-Gnutella06 & computer & \numprint{8717} & \numprint{31525} & \numprint{10} & \numprint{0.0000} & \numprint{0.01}\% &\numprint{0.01}\% & \numprint{0.00}\% & \numprint{0.00}\%\\
p2p-Gnutella05 & computer & \numprint{8842} & \numprint{31837} & \numprint{9} & \numprint{0.0001} & \numprint{0.02}\% &\numprint{0.01}\% & \numprint{0.01}\% & \numprint{0.00}\%\\
\midrule
\end{tabular}

  \end{adjustbox}
\end{table}
}

\end{document}